\documentclass[12pt]{article}
\usepackage[T1]{fontenc}
\usepackage{microtype}
\usepackage{libertine}
\usepackage[T1]{fontenc}
\usepackage[italic]{mathastext}

\usepackage{geometry}
\DeclareMathAlphabet{\mathpzc}{OT1}{pzc}{m}{it}

\usepackage{amsfonts}
\usepackage{amsmath}
\usepackage{amssymb}
\usepackage{graphicx}
\usepackage[table]{xcolor}
\usepackage[toc,page]{appendix}
\usepackage{multirow}
\usepackage{ulem}
\usepackage{setspace}
\usepackage[dvipsnames]{pstricks}
\usepackage{pst-node}
\usepackage{pst-slpe}
\usepackage[authoryear]{natbib}
\usepackage{dcolumn}
\usepackage{tabularx}
\usepackage{booktabs}
\usepackage{caption}
\usepackage{rotating}
\usepackage{threeparttable}
\usepackage{threeparttablex}
\usepackage{longtable}
\usepackage{pdflscape}
\usepackage[colorlinks=true,linkcolor=blue]{hyperref}
\usepackage[nameinlink,capitalize,noabbrev]{cleveref}
\usepackage{xr-hyper}
\hypersetup{citecolor=black, linkcolor=black, urlcolor=black, filecolor=black,colorlinks=true}

\usepackage{numprint}
\usepackage{eurosym}
\usepackage{textcomp}
\usepackage{stackengine}
\usepackage{subcaption}
\usepackage{etoolbox}
\usepackage{siunitx}
\usepackage{accents}
\usepackage{xr}
\usepackage{makecell}
\usepackage{authblk}
\usepackage{bbm}
\usepackage{xfrac}

\DeclareFontFamily{U}{mathx}{\hyphenchar\font45}
\DeclareFontShape{U}{mathx}{m}{n}{<-> mathx10}{}
\DeclareSymbolFont{mathx}{U}{mathx}{m}{n}
\DeclareMathAccent{\widebar}{0}{mathx}{"73}

\robustify\itshape
\robustify\bfseries
\setstackEOL{\#}
\setstackgap{L}{12pt}
\npdecimalsign{.}
\npthousandsep{,}
\npthousandthpartsep{\,}
\newcolumntype{d}[1]{D.{\mbox{.}}{#1}}

\newcolumntype{d}[1]{D.{\mbox{.}}{#1}}
\newcolumntype{R}{>{\raggedleft\arraybackslash}X} 
\newtheorem{proposition}{Proposition}



\begin{document}
	
	\title{The Role of Risk Sharing in Attenuating Business Cycles Within Currency Unions}
	\author[1]{Alberto Pavia\thanks{alberto.paviasoto@kuleuven.be, Naamsestraat 69, 3000 Leuven, Belgium. }}
	\author[1]{Christian Proebsting \thanks{christian.probsting@kuleuven.be, Naamsestraat 69, 3000 Leuven, Belgium. \\ We would like to thank Chris Boehm, Christophe Blot, Giancarlo Corsetti, Luca Dedola, Yuriy Gorodnichenko, Mathias Hoffmann, Raphaël Huleux, Gernot Müller, Emi Nakamura, Fabrizio Perri, Anna Rogantini Pico, Bent Sørensen, Jón Steinsson, Burak Uras, as well as numerous seminar and conference participants for useful comments. We gratefully acknowledge funding from the Research Foundation – Flanders (FWO), project no. G040923N and fellowship no. 1132725N.}}
	\date{\today \vspace{1cm}}  
	\affil[1]{KU Leuven}
	\maketitle
	
	\large 
	\vspace*{-1cm}
	\begin{abstract}
		The United States is a currency union where multiple risk-sharing mechanisms--- migration, fiscal transfers, income diversification and credit markets---buffer consumption from local income fluctuations. We show that risk sharing not only directly smooths consumption but also indirectly stabilizes income by dampening the local multiplier. Combining causal estimates from regional military buildups with a multi-region quantitative model, we find that current levels of risk-sharing cut state-level consumption volatility by a factor of four. Crucially, the indirect stabilization of income accounts for nearly half of this effect, implying substantially larger benefits from integration than conventional measures suggest.
	\end{abstract}
	
	\small
	\textbf{JEL Codes}: F44, F45, E32, F15, F41, F36.
	
	\textbf{Keywords}: risk sharing, currency unions, consumption volatility.

	\newpage
	\section{Introduction}	
	In the 1980s, Massachusetts found itself at the center of a technological boom. The rise of the mini-computer industry sent local employment and incomes soaring, only for the state to slip into a sharp downturn when personal computers undercut the mini-computer business model. Episodes like this are hardly unique across the United States: The Dakotas rose and fell with the fortunes of shale; and manufacturing centers across the Midwest have cycled through periods of expansion and contraction for decades. These regional swings leave clear marks on employment and income, yet households are partially shielded from their full force. An array of risk-sharing mechanisms—such as migration toward stronger labor markets, federal programs that shift resources when conditions weaken, and income streams extending beyond the immediate region—act as a form of insurance, allowing consumption to remain steadier than local incomes alone would imply.	
	
	In this paper, we argue that this insurance view is incomplete. We show that risk sharing not only smooths consumption relative to income; it also \emph{prevents income from falling as much in the first place}. We call this the indirect effect of risk sharing. This channel is quantitatively large, accounting for half of the total stabilization in consumption volatility attributable to risk-sharing mechanisms in the United States. Accounting for this channel implies that the stabilizing benefits of risk sharing are substantially larger than previously understood.
	
	The logic rests on the local Keynesian multiplier. When a region suffers a negative shock, local businesses see revenues fall, workers face lower earnings, and households cut their spending. This decline in consumption triggers a second round of cutbacks, a feedback loop that deepens the recession. Risk-sharing mechanisms cut through this cycle by supporting households' income and spending. By breaking the feedback loop between falling income and falling consumption, risk sharing does not just insure against the business cycle; it dampens the local multiplier and actively reduces the severity of the downturn.
	
	Our papers quantifies the role of cross-regional risk sharing mechanisms in attenuating local business cycles. Previous work has relied on variance decompositions of unconditional income fluctuations to quantify the amount of cross-border risk sharing \citep[see][ (henceforth: ASY) and the subsequent literature]{asdrubali1996channels}. We show that such methods implicitly treat income volatility as exogenous to the level of risk sharing and cannot capture how local income itself depends on local consumption through the local Keynesian multiplier. Instead, we use an identified local business cycle shock to jointly estimate how both consumption and income respond to such a shock. This setup is informative about both the direct and indirect benefits of risk sharing: More risk sharing directly mutes the response of consumption for a given path of income, but it also reduces the response of income itself, thereby indirectly smoothing consumption. 
	
	Specifically, we exploit regional variation in military buildups to isolate exogenous shifts in external demand across U.S. states \citep{nakamura2014fiscal}. We estimate that an increase in external demand of 1 percent of GDP raises a state’s GDP by 1.31 percent upon impact, while consumption rises by half as much (0.62 percent), a ``pass-through'' of income into consumption fluctuations by 47 percent ($=0.62/1.31$). While informative, these estimates by themselves cannot directly reveal how much consumption volatility would rise in the absence of risk sharing. A naïve back-of-the-envelope calculation (assuming the volatility of income does not depend on the level of risk sharing) would suggest that without risk sharing, consumption would simply track income one-for-one. Based on our estimates, this implies consumption volatility would roughly double.
	
	Our findings show that such a calculation drastically understates the gains from risk sharing. In the absence of risk sharing, not only would consumption move one-for-one with income, but the size of the fluctuations in income would be larger due to the loss of stabilization. To precisely quantify this, we build a quantitative multi-region model that incorporates trade, migration, fiscal transfers, and financial markets. The model is disciplined by our causal estimates. Our counterfactual analysis shows that without risk sharing, the local multiplier would nearly double to 2.52. Combining the direct and indirect effects, consumption volatility would therefore almost quadruple, with the indirect effects accounting for nearly half of the effect. 
	
	Digging deeper, we decompose the overall reduction in consumption volatility into four cross-state risk-sharing channels: migration, income diversification through capital markets, fiscal transfers and borrowing \& lending through credit markets (ASY, \cite{parsley2021risk}). We discipline the strength of each channel in our model by regressing alternative outcome variables on the external demand shock: population changes identify the migration channel, differences between GDP and personal income responses capture income diversification, gaps between personal and disposable income reveal fiscal transfers, and divergences between disposable income and consumption reflect borrowing and lending.
	
	This exercise yields two central findings: First, our empirical estimates challenge the conventional view that almost half of income fluctuations across U.S. states are shared through capital markets.\footnote{ASY estimate that 39 percent of fluctuations in states' GDP are smoothed by capital markets. Subsequent studies even larger numbers \citep[see e.g.][]{parsley2021risk,kohler2023risk}.} Our estimates reveal that the direct benefits of income diversification through capital markets are about three times smaller. We show that this discrepancy arises because unconditional state-level fluctuations in GDP are primarily driven by capital income---which is easier to diversify than labor income, whereas the demand-driven shocks we measure have a more balanced incidence on labor and capital income. This finding also rationalizes why we find a larger pass-through of income fluctuations into consumption fluctuations compared to ASY (47 percent vs. 25 percent). 
	
	Second, our structural model suggests that not all risk-sharing mechanisms contribute equally once their interactions are taken into account. Moving from no risk-sharing to the observed level of risk sharing cuts the consumption multiplier by a factor of four. Fiscal transfers account for more than half of this reduction, whereas income diversification through capital markets and migration contribute about a fifth each. The model highlights economically meaningful substitution effects across the various risk sharing channels. For instance, cutting back the tax and transfer system would create stronger incentives for workers to move in response to a local fall in income, partially compensating for the initial loss in risk sharing. This means that a currency union does not necessarily need all four risk-sharing channels, as already implementing a single risk-sharing channel can go a long way in smoothing out output and consumption fluctuation
	
	Finally, we show that U.S. states differ in how much they benefit from the current level of risk sharing. A central factor is how much a state is integrated in the network of interregional trade. States that trade more extensively with the rest of the country derive smaller additional gains from risk sharing, since trade itself dampens the local multiplier by reducing the dependence of local incomes on local spending.
	
	Our findings have important implications for the design of currency unions, particularly the euro area. They suggest that the benefits of creating integrated risk-sharing mechanisms, such as a common fiscal capacity, are substantially larger than previously thought. Such mechanisms do not merely provide insurance; they fundamentally alter the business cycle dynamics of member states by dampening the effect of local shocks at their source. 
	
	The remainder of the paper is structured as follows. After reviewing the literature, we present a simple analytical model to formalize the direct and indirect effects of risk sharing. Section \ref{sec:empirical_results} details our empirical strategy and results. Section \ref{sec:model} introduces the quantitative multi-region model used to perform the counterfactual analysis. Section \ref{sec:model_solution} outlines calibration and estimation, and Section \ref{sec:counterfactuals} describes our quantitative results. The final section concludes.

	\paragraph{Related Literature}
	This paper provides a structural interpretation of risk sharing, showing that its stabilizing benefits are almost twice as large as previously understood. To do so, we link two distinct literatures: one on empirical risk sharing and another on regional fiscal multipliers.  The seminal accounting framework of the risk-sharing literature is structurally agnostic; it treats income fluctuations as exogenous and cannot quantify the general equilibrium effects of stabilization policies. Conversely, the regional multiplier literature provides the necessary causal identification but has not historically examined how risk sharing endogenously determines the size of the multiplier. We combine these approaches to show that risk sharing does not merely smooth consumption but actively dampens the local multiplier, providing a direct link to the theory of optimum currency areas.
	
	Our work engages the vast empirical literature on risk sharing pioneered by \cite{cochrane1991simple}, \cite{mace1991full}, and \cite{asdrubali1996channels}\footnote{See, among others, \cite{sorensen1998international, becker2006intra, hoffmann2008lack, flood2012international} and \cite{kohler2023risk}.}, but we depart from its methodological reliance on unconditional variance decompositions. Our first contribution is to adopt a causal identification strategy that isolates the response to demand shocks using regional variation in military spending. This allows us to provide the first causal estimates of both the direct consumption-smoothing channel and the indirect output-stabilizing channel of risk sharing.
	
	Our second contribution is structural.  While the international business cycle literature has long employed general equilibrium models to study risk sharing and shock transmission (e.g.,\citealp{backus1992international,corsetti2008international}), prior structural work on currency unions has largely focused on modeling single risk-sharing channels in isolation (e.g., \citealp{house2025quantifying} for labor mobility; \citealp{evers2015fiscal} for fiscal transfers). In contrast, our general equilibrium framework incorporates migration, capital markets, fiscal transfers, and credit markets simultaneously. This approach, disciplined by our causal estimates, allows us to quantify the interdependencies between channels. We find that these mechanisms act as substitutes; consequently, the effectiveness of any single policy (such as expanding fiscal capacity) depends critically on the presence or absence of others (such as integrated credit markets) This structural approach provides quantitative guidance for policy debates on the design of currency unions in the United States and Europe.
	
	Our findings also provide a new perspective on the modern literature on cross-regional fiscal multipliers \citep{nakamura2014fiscal, chodorow2019geographic}. This literature is influential because regional estimates are considered robust to confounding aggregate policy responses, such as monetary policy or federal taxes, but it is still an open question how regional multipliers map into aggregate multipliers. \cite{chodorow2019geographic} concludes that regional multipliers provide a rough lower bound for the aggregate, no-monetary-policy-response multiplier. In a similar vein, \cite{nakamura2014fiscal} argue that their large regional multiplier estimates favor models in which demand shocks can have large effects on output if monetary policy is sufficiently accommodative. Our analysis suggests an alternative interpretation: that the magnitude of these regional multipliers is highly sensitive to the presence of risk-sharing mechanisms. Even neoclassical models that imply small aggregate multipliers are consistent with potentially large regional multipliers, as long as regional risk sharing is limited. 
	
	In a complementary contribution, our work speaks to the recent literature that uses regional estimates to identify micro parameters or other general equilibrium effects \citep{guren2021we, wolf2023missing}. In the spirit of this work, we show that our framework can be used to map the size of the regional spending multiplier to the quantitative importance of our indirect risk-sharing channel, providing a new tool for interpreting regional variation.
	
	Finally, our work relates to the broader literature on optimum currency areas  \citep{friedman1953case,mundell1961theory}. Our contribution is twofold. First, we provide a holistic framework that quantifies several of the most prominent criteria proposed to stabilize consumption fluctuations across U.S. states. Second, we highlight important interdependencies across these criteria, which are often discussed separately in the literature. For example, \citet{mckinnon1963optimum} emphasized trade openness as a prerequisite for an optimal currency area. We show that it acts as a substitute for risk-sharing because it weakens the feedback loop of the Keynesian local multiplier by making local income less dependent on local demand.\footnote{See \cite{farhi2014labor} for pointing out interaction effects between labor mobility and trade openness.} 
	
	\section{Risk Sharing and Open-Economy Multipliers} 
	\label{sec:simple_model}
	Consider a small open economy that faces a temporary and unexpected increase in nominal federal government spending: $d G^{nom}_{t}$. Conceptually, we can decompose the response of nominal consumption into two components: the elasticity of nominal consumption to nominal output, $ \frac{d \ln {C}^{nom}_{t}}{ d \ln {Y}^{nom}_{t}}$, and the change in output, $d Y^{nom}_{t}$. The log change in nominal consumption for an increase in nominal federal government spending corresponding to 1 percent of output (normalized to $Y^{nom} = 1$ such that $d \ln Y^{nom}_t = d Y^{nom}_t$) is given by 
	\begin{align}
		\left.\frac{d\ln C^{nom}_t}{d G^{nom}_t}\right|_{Y^{nom}=1} =  \frac{d \ln C^{nom}_{t}}{d \ln Y^{nom}_{t}} \cdot  \frac{d Y^{nom}_{t}}{d G^{nom}_{t}}. 
	\end{align} 
	Both of these terms have received substantial attention in the literature. The first term, $ \frac{d \ln C^{nom}_{t}}{d \ln Y^{nom}_{t}}$, is a central object to the empirical international risk-sharing literature. ASY estimate this elasticity by regressing log changes in nominal consumption on log changes in nominal output across U.S. states. Following this literature, we denote this elasticity by $\beta^C_t\equiv  \frac{d \ln C^{nom}_{t}}{d \ln Y^{nom}_{t}}$ because it captures the share of risk that passes-through to consumption after accounting for all risk-sharing channels. Under log utility, perfect risk sharing requires that per-capita nominal consumption grows at equal rates in the small open economy and the rest of the world \citep{backus1993consumption}. This implies that the growth rate of the country's nominal consumption is independent of its nominal output growth: $\beta^C_t = 0$. With imperfect risk sharing, $\beta^C_t>0$, and in the extreme case of $\beta^C_t = 1$, income fluctuations translate one-for-one into consumption fluctuations.
	
	The second term, $\frac{d Y^{nom}_t}{d G^{nom}_t}$, is the open-economy relative multiplier. This follows \cite{nakamura2014fiscal}, who estimate this multiplier by exploiting variation in external demand (federal government spending) across U.S. states. We denote this term by $m^Y_t$.

	Combining these definitions, we write the consumption multiplier, $m^C_t$, as
	\begin{align}
		\label{eq:response_C}
		m_t^C \equiv \frac{d \ln C^{nom}_{t}}{d G^{nom}_{t}} =  \beta^C_t \cdot m_t^Y. 
	\end{align}
	The multiplier $m^Y_t$ indicates the percent change in income for a shock to external demand corresponding to 1 percent of output. The elasticity $\beta^C_t$ governs how this change in income translates into a percent change in consumption. 
	
	While the literature typically studies risk sharing and fiscal multipliers in isolation, we show that they are endogenously related. Specifically, the multiplier is an increasing function of the share of unsmoothed income $\beta^C_t$, i.e. $m^Y_t(\beta^C_t)$ with $\frac{\partial m^Y_t}{\partial \beta^C_t}>0$. This interrelationship implies that greater risk sharing (a lower $\beta^C_t$) confers a dual benefit: It directly dampens the response of consumption to a \textit{given} change increase in income, but it also simultaneously reduces the sensitivity of output to the initial increase in external demand (a lower $m^Y_t$). We now formalize this argument within a standard open-economy model.

	\subsection{A Simple Open-Economy Model} 
	
	We consider a continuum of ex-ante identical small open economies that belong to a currency union \citep{gali2005monetary}. Each economy experiences idiosyncratic fluctuations in federal government demand for its goods, while aggregate federal government spending at the currency union level remains constant. Within this framework, we establish a formal link between the open-economy relative multiplier, $m_t^Y$, and the share of unsmoothed income fluctuations, $\beta^C_t$.
	
	This relationship is independent of several model features that are known to influence the closed-economy multiplier, such as the degree of nominal price or wage rigidity, monetary policy, decreasing returns to scale in the production function, input market segmentation, and the elasticity of labor supply. Instead, the multiplier depends critically on the economy's openness in both financial and goods markets. We therefore focus our exposition on these two features and present the remaining model components alongside the quantitative model in section \ref{sec:model}.\footnote{While these other features do influence the multiplier, they do so only through their effect on $\beta^C_t$.}

	\subsubsection{Trade in Financial Assets}
	In this simple model, we focus on a single, generic risk sharing mechanism---an assumption we later relax in the quantitative model by incorporating multiple risk-sharing channels that can be mapped directly to the data. We assume that households can trade state-contingent bonds, but these transactions are subject to a tax. This friction, in the spirit of \cite{devereux2014capital}, allowing us to parameterize the economy anywhere between financial autarky and complete markets.
	
	More formally, let $B_i(s^t,s_{t+1})$ denote the quantity of state-contingent bonds purchased by households in country $i$ after history $s^t$ that pay one unit of the union's currency in state $s^{t+1}$.\footnote{We employ the standard notation for models that explicitly account for state-contingent bonds: In each period $t$, the economy experiences one of finitely many events $s \in S$. The transition probability from state $s$ to $s'$ follows a Markov chain denoted by $\pi(s'|s)$. We use $s^t=(s_0,s_1,...,s_t)$ to represent the history of events through period $t$, with probability $\pi(s^t)$ as of period 0.} The price of these bonds to households in the small open economy is $(1+t_i(s^t,s_{t+1})) Q(s^t,s_{t+1})$, where $Q(s^t,s_{t+1})$ is the net price and $t_i(s^t,s_{t+1})$ is a (potentially negative) tax levied by the federal government. This tax rate depends on both the specific bond to which it applies to and the country. 
	
	Beyond bond income, households receive nominal income from output production, $P_i^Y(s^t) Y_i(s^t)$, where $P_i^Y(s^t)$ is the nominal price of the economy's output, and purchase consumption goods, $C_i(s^t)$, at price $P_i(s^t)$. We remain agnostic about the production process, but assume a fixed capital stock that does not require investment. Finally, households in the economy need to pay a lump-sum federal tax $T(s^t)$ that is common across countries. 
	
	The household's budget constraint is
	\begin{align}
	P_i(s^t) C_i(s^t) + \sum_{s^{t+1}} (1+t_i(s^t,s_{t+1})) Q(s^t,s_{t+1}) B_i(s^t,s_{t+1}) = P_i^Y(s^t) Y_i(s^t) - T(s^t) + B_i(s^{t-1},s_t).
	\end{align} 
	We assume households have log utility over consumption, and consumption and labor are separable. Then, the optimal choice of bond positions yields the following international risk-sharing condition for any country pair $i$ and $j$:
	\begin{align}
		\label{eq:BackusSmith_imperfect}
		\frac{C^{nom}_i(s^{t+1})}{C^{nom}_i(s^{t})}	(1+t_i(s^{t},s_{t+1}))  &= \frac{C^{nom}_j(s^{t+1})}{C^{nom}_j(s^{t})}	(1+t_j(s^{t},s_{t+1})), 
	\end{align}
	where $C_i^{nom}(s^t) \equiv P_i(s^t) C_i(s^t)$ is nominal consumption. Absent the tax on bonds, nominal consumption growth would be equalized across countries in all possible states, implying full risk sharing \citep{backus1993consumption}. Consumption growth would not depend on countries' income growth: Households would fully insure themselves against income fluctuations by buying bonds that pay out in low-income states and selling bonds that pay out in high-income states.
	
	The taxes on bonds break this perfect risk sharing condition. We specify country $i$'s tax rate for bond $B(s^t,s_{t+1})$ as a decreasing function of that countries' income growth:
	\begin{align}
	1 + t_i(s^t,s_{t+1}) = \left( \frac{Y^{nom}_i(s^{t})}{Y^{nom}_i(s^{t+1})} \right)^{1-\lambda},		
	\end{align}
	where $Y^{nom}_i(s^{t}) \equiv P^Y_i(s^{t})Y_i(s^t)$ is nominal output and $\lambda \in [0,1]$ governs the degree of curvature of the tax function. This specification reduces households' incentives to engage in self-insurance: For households in country $i$, the tax raises the price of bonds that pay out in states when country $i$ has relatively low income, making it more expensive for these households to self insure. Conversely, the price of bonds that pay out in high-income states have a low, potentially negative tax rate. Substituting this tax function into the risk-sharing condition (\ref{eq:BackusSmith_imperfect}) yields
	\begin{align}
		\label{eq:risk_sharing_lambda} 
		C^{nom}_{i,t}= \kappa_i \kappa_t \left(Y^{nom}_{i,t}  \right)^{1-\lambda}, 
	\end{align}
	where $\kappa_i$ is a country-specific constant and $\kappa_t$ is a union-wide variable.\footnote{The Appendix spells out $\kappa_i$ and $\kappa_t$.} align (\ref{eq:risk_sharing_lambda}) demonstrates how $\lambda$ determines the degree of risk sharing: If $\lambda=0$, trade is balanced period-by-period, with nominal consumption equaling nominal output ($C^{nom}_{i,t} = Y^{nom}_{i,t}$), representing financial autarky. If $\lambda=1$, the tax is zero and markets are complete, resulting in full risk sharing across countries ($C^{nom}_{i,t} = C^{nom}_t$). More generally,
	\begin{align}
	\beta^C_t \equiv  \frac{d \ln C^{nom}_t}{d \ln Y^{nom}_t} = 1-\lambda,	
	\end{align}
	so that the share of income fluctuations left unsmoothed is constant over time.\footnote{Given the small open economy assumption, changes in nominal output driven by idiosyncratic shocks to country $i$ have no effect on $\kappa_t$.} The parameter $\lambda$ can therefore be interpreted as the degree of international financial integration: higher values of $\lambda$ imply that a larger share of idiosyncratic income shocks is diversified internationally.

	\subsubsection{Trade in Goods} 
	We model international trade in goods following \cite{armington1969theory}, as is standard in the open-economy literature. 
	Consumption in this economy consists of a CES aggregate of domestically produced goods, $C^{dom}_{i,t}$, and imports, $C^{imp}_{i,t}$, with elasticity of substitution $\psi\geq 0$:
	\begin{align}
		C_{i,t} = \left( \left( \omega \right)^{\frac{1}{\psi}} \left( C^{dom}_{i,t} \right)^{\frac{\psi-1}{\psi}} + (1-\omega)^{\frac{1}{\psi}} \left(C^{imp}_{i,t} \right)^{\frac{\psi-1}{\psi}} \right)^{\frac{\psi}{\psi-1}}.  
	\end{align}%
	where $\omega$ represents the home bias parameter. 
	The nominal price of domestic goods is $P^Y_{i,t}$, while the nominal price of imported goods is $P^{imp}_{i,t}$. Optimal demand for domestic and imported goods is:
	\begin{align}
		C^{dom}_{i,t} = \omega C_{i,t} \left( \frac{P^Y_{i,t}}{P_{i,t}} \right)^{-\psi} \qquad \text{and} \qquad 
		C^{imp}_{i,t} = (1-\omega) C_{i,t} \left( \frac{P^{imp}_{i,t}}{P_{i,t}} \right)^{-\psi}.
	\end{align}
	The rest of the world demands the SOE's goods according to a similar demand schedule:
	\begin{align}
		C^{exp}_{i,t} = \bar{C}^{exp,nom}_{i} \left(P^Y_{i,t}  \right)^{-\psi},
	\end{align}
	where $\bar{C}^{exp,nom}_{i}$ is nominal exports in the non-stochastic steady state. Finally, the good produced by the SOE can also be purchased by the federal government. The market clearing condition therefore reads
	\begin{align}
	Y_t = C^{dom}_{i,t} + G_{i,t} + C^{exp}_{i,t}.	
	\end{align}
	By substituting the optimal demand for domestic goods into this market-clearing condition, we derive a relationship between the multiplier and the share of unsmoothed income fluctuations, as formalized in Proposition \ref{prop:Prop1}.
	\begin{proposition}[Risk Sharing and Output Multipliers]\label{prop:Prop1}
		Consider the model discussed in this section. The share of income fluctuations that remain unsmoothed is given by
		\begin{align} 
			\beta^C_{i,t} \equiv \frac{d \ln C^{nom}_{i,t}}{d \ln Y^{nom}_{i,t}} = 1-\lambda,
		\end{align}
		where $\lambda \in [0,1]$ is the index of financial market integration. Assuming a unit Armington elasticity, $\psi=1$, the open-economy relative output multiplier is increasing in $\beta_t^C$:
		\begin{align}
			m_{i,t}^Y \equiv \frac{d Y^{nom}_{i,t}}{d G^{nom}_{i,t}} = \frac{1}{1-\omega \bar{\mathpzc{c}} \beta_{i,t}^C} = \frac{1}{1-\omega \bar{\mathpzc{c}} (1-\lambda)},
		\end{align}
		where $\omega$ is the share of final expenditure that falls on domestic goods, and $\bar{\mathpzc{c}} \equiv \frac{\bar{C}^{nom}}{\bar{Y}^{nom}}$ is the share of nominal consumption in nominal output. 
		
	\end{proposition}
	\textbf{Proof:} Substituting the demand for $C^{dom}_{i,t}$ and for $C^{exp}_{i,t}$ into the market-clearing condition with $\psi=1$ and multiplying by $P^Y_{i,t}$ yields:
\begin{align}
	Y^{nom}_{i,t} = \omega C^{nom}_{i,t} + G^{nom}_{i,t} + \bar{C}^{exp,nom}_{i}.
	\label{eq:Y_nom}
\end{align}

Differentiating with respect to \(G^{nom}_{i,t}\) yields
\begin{align}
	\frac{d Y^{nom}_{i,t}}{d G^{nom}_{i,t}}
	= \omega \frac{\partial C^{nom}_{i,t}}{\partial Y^{nom}_{i,t}}
	\frac{d Y^{nom}_{i,t}}{d G^{nom}_{i,t}} + 1.
	\label{eq:diff}
\end{align}

Rearranging and denoting the open‐economy multiplier by
\(m_{i,t}^Y \equiv \frac{d Y^{nom}_{i,t}}{d G^{nom}_{i,t}}\) yields:
\begin{align}
	m_{i,t}^Y = \left( 1 - \omega \bar{\mathpzc{c}} \beta^C_{i,t} \right)^{-1},
	\label{eq:multiplier}
\end{align}

where
\(\beta^C_{i,t} \equiv \frac{d \ln C^{nom}_{i,t}}{d \ln Y^{nom}_{i,t}}\)
is the elasticity of consumption to output.

\(\blacksquare\)

	Proposition 1 establishes that risk sharing reduces the open-economy relative multiplier. The intuition is as follows: consider a \$1 increase in external demand. Under complete financial markets ($\beta^C=0$), consumption remains unchanged and output rises by exactly \$1, yielding $m^Y|_{\beta^C=0} = 1$. Under financial autarky ($\beta^C=1$), the output increase translates one-for-one into higher consumption (assuming that the share of federal military spending in the SOE's GDP is negligible, $\bar{\mathpzc{c}}  \approx 1$). A share $\omega$ of this consumption increase falls on domestic goods, further stimulating output and consumption through a multiplier effect. The total impact exceeds that under complete markets:
	\begin{align}
	\left.m^Y\right\vert_{\beta^C=1} = 1 + \omega + \omega^2 + ... = \frac{1}{1-\omega} \geq 1 = \left.m^Y\right\vert_{\beta^C=0}.	
	\end{align}
	
	The continuous, blue line in Panel (a) of Figure \ref{fig:simpleModel} illustrates this nonlinear relationship between risk sharing ($\lambda = 1-\beta^C$) and the multiplier ($m^Y$). Greater risk sharing (higher $\lambda$) lowers the multiplier substantially. Quantitatively, starting from a multiplier equal to $m=4.5$ under zero risk sharing ($\lambda=0$), the multiplier falls to $m^Y=2.35$ for $\lambda = 0.25$ and to $m^Y=1.65$ for $\lambda=0.5$. Under perfect risk sharing ($\lambda=1$), the multiplier equals 1.
	
	Home bias ($\omega$) plays a crucial role in this relationship, with higher home bias amplifying the feedback loop. For our analysis, we set $\omega = 0.78$, which corresponds to the average home bias across U.S. states \citep[see below and][]{rodriguez2025trade}.\footnote{We implicitly assume that government spending falls entirely on the good produced by the SOE, as in \cite{nakamura2014fiscal}. Alternatively, we could have assumed that a fraction $1-\omega$ of government spending in the SOE falls on imported goods (as we assume for private consumption). In that case, the expression for the multiplier needs to be multiplied by $\omega$ because only a fraction $\omega$ of government spending actually falls on SOE's goods.} The Armington elasticity is equally important. The assumed unit elasticity simplifies the analysis by ensuring constant expenditure shares on domestic goods. A higher Armington elasticity would amplify expenditure switching and cause the expenditure share on domestic goods to decline, thereby weakening the feedback loop and the relationship between risk sharing and the multiplier. Empirical estimates of the Armington elasticity in the literature range from 5–8 in the long run \citep{broda2006groundnuts} to values potentially below unity at business cycle frequencies \citep{boehm2023long}. For comparison, \cite{nakamura2014fiscal} employ a DSGE model with complete financial markets (which corresponds to the case $\lambda = 1$) and an Armington elasticity of 2, yielding multipliers below unity: They report 0.83 in a model with sticky prices (which mutes expenditure switching) and 0.43 in a model with flexible prices.
	
	Having established the link between risk sharing and the open-economy output multiplier, we can now derive the consumption multiplier and decompose the effects of risk sharing on the consumption multiplier into a direct and an indirect effect:
	
	\begin{proposition}[Risk Sharing and Consumption Multipliers]\label{prop:Prop2}
		The nominal consumption multiplier $m^c_{i,t}$ is 
		\begin{align}
			m^c_{i,t} = \beta^C_{i,t} \cdot m^Y_{i,t} = (1-\lambda) \cdot \frac{1}{1-\omega \bar{\mathpzc{c}} (1-\lambda)}.
		\end{align}
		A marginal increase in risk sharing (i.e., an increase in $\lambda$) reduces the consumption multiplier through a direct effect ($\frac{\partial \beta^C_{i,t}}{\partial \lambda} \cdot m_{i,t}^Y$) and an indirect effect ($\frac{\partial m_{i,t}^Y}{\partial \lambda} \cdot \beta^C_{i,t}$). The indirect effect's share in the total marginal effect is:
		\begin{align}
			\mathcal{S}_{i,t}^{indirect} \equiv \frac{\frac{\partial m_{i,t}^Y}{\partial \lambda} \cdot \beta^C_{i,t}}{\frac{d m^c_{i,t}}{d \lambda}} = \omega \bar{\mathpzc{c}} (1-\lambda). 
		\end{align}
	\end{proposition}
	\textbf{Proof:} The nominal consumption multiplier follows directly from $m^c_{i,t} = \beta^C_{i,t} \cdot m_{i,t}^Y$ and the expressions for $\beta^C_{i,t}$ and $m_{i,t}^Y$ in Proposition 1. Using the product rule to decompose the marginal effect of a change in $\lambda$ on the consumption multiplier yields
	\begin{align} 
	\frac{d m^c_{i,t}}{d\lambda} = \frac{\partial \beta^C_{i,t}}{\partial \lambda} \cdot m_{i,t}^Y + \frac{\partial m_{i,t}^Y}{\partial \lambda} \cdot \beta^C_{i,t} = -m_{i,t}^Y - \omega \bar{\mathpzc{c}} \left(m^Y_{i,t}\right)^2 \beta^C_{i,t}.
	\end{align} 
	Then, the share of the indirect effect in the total marginal effect is
	\begin{align} 
	\mathcal{S}_{i,t}^{indirect} \equiv \frac{\frac{\partial m^Y_{i,t}}{\partial \lambda} \cdot \beta^C_{i,t}}{\frac{d m^c_{i,t}}{d \lambda}} = \frac{\omega \bar{\mathpzc{c}}\beta^C_{i,t}}{ \left(m^Y_{i,t}\right)^{-1} +\omega \bar{\mathpzc{c}} \beta^C_{i,t}}.
	\end{align}
	Inserting the expressions for $\beta^C_{i,t}$ and $m^Y_{i,t}$ yields the result. 
	$\blacksquare$  
	
	The dashed red line in Panel (a) of Figure \ref{fig:simpleModel} depicts the consumption multiplier $ m^C = (1-\lambda) \cdot \frac{1}{1-\omega (1-\lambda)}$. Under financial autarky, consumption and output multipliers coincide since $C^{nom}_{i,t} = Y^{nom}_{i,t}$ at all times by definition. As financial integration increases, the consumption multiplier declines through two channels: risk sharing directly reduces consumption volatility for a given level of output volatility, and it indirectly reduces output volatility itself (as shown by the declining blue line).
	
	Panel (b) of Figure \ref{fig:simpleModel} illustrates the indirect effect's contribution to the total effect, $	\mathcal{S}^{indirect}$. For low levels of financial market integration, the indirect effect of reducing output volatility dominates. Modest changes in financial market integration can substantially reduce consumption volatility---the consumption multiplier falls by a factor of five when risk sharing increases from $\lambda=0$ to $\lambda=0.5$. The relevance of the indirect effect is directly related to the degree of home bias: greater home bias amplifies the feedback loop as output becomes more dependent on domestic consumption. Consequently, the stabilization gains from financial integration are larger for economies that are relatively closed to international trade in goods.
	
	\begin{figure}
		\centering
		\begin{subfigure}[b]{0.48\linewidth}
			\centering
			\includegraphics[width=\linewidth]{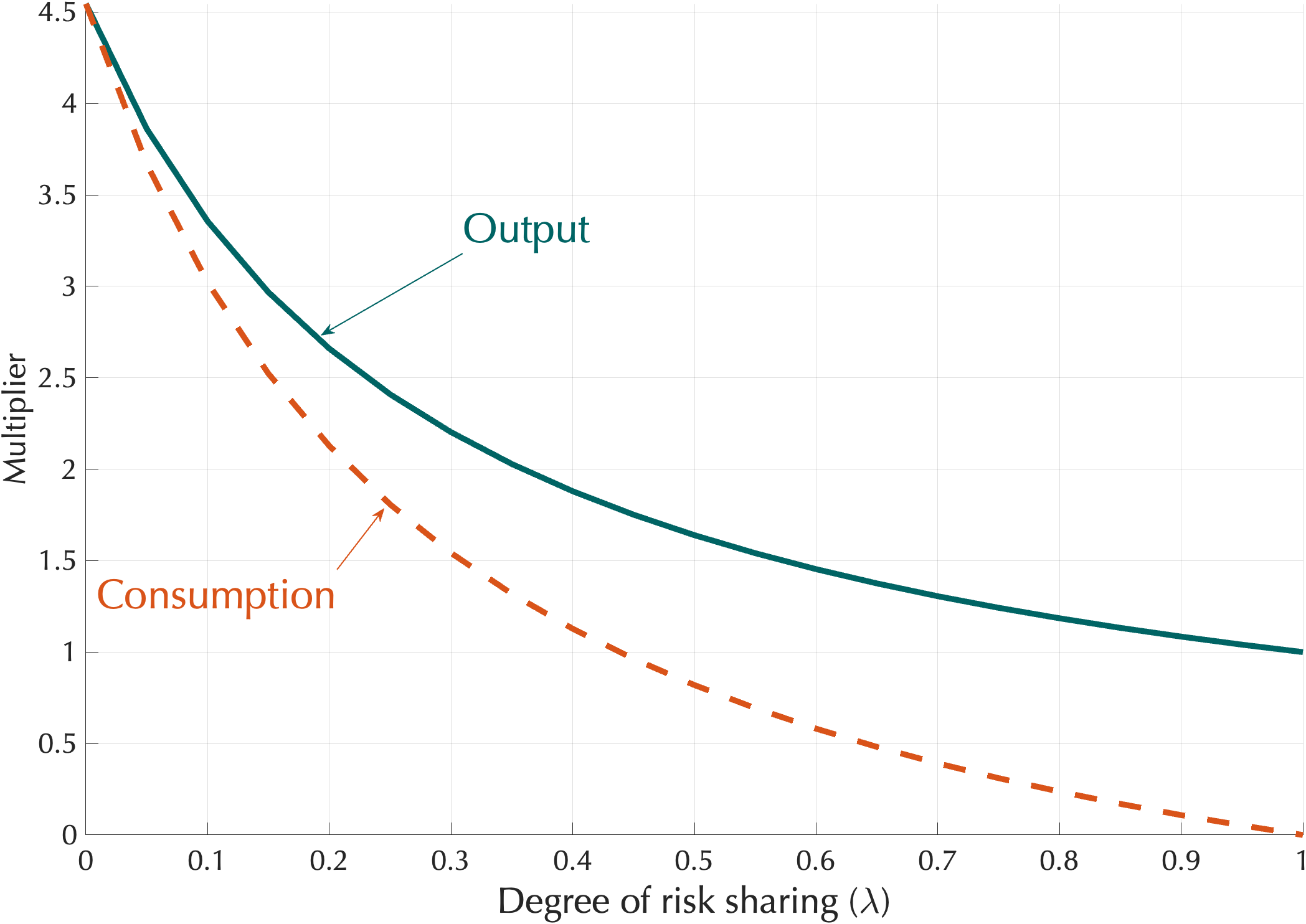}
			\caption{Open-economy relative multipliers}
		\end{subfigure}
		\hfill
		\begin{subfigure}[b]{0.48\linewidth}
			\centering
			\includegraphics[width=\linewidth]{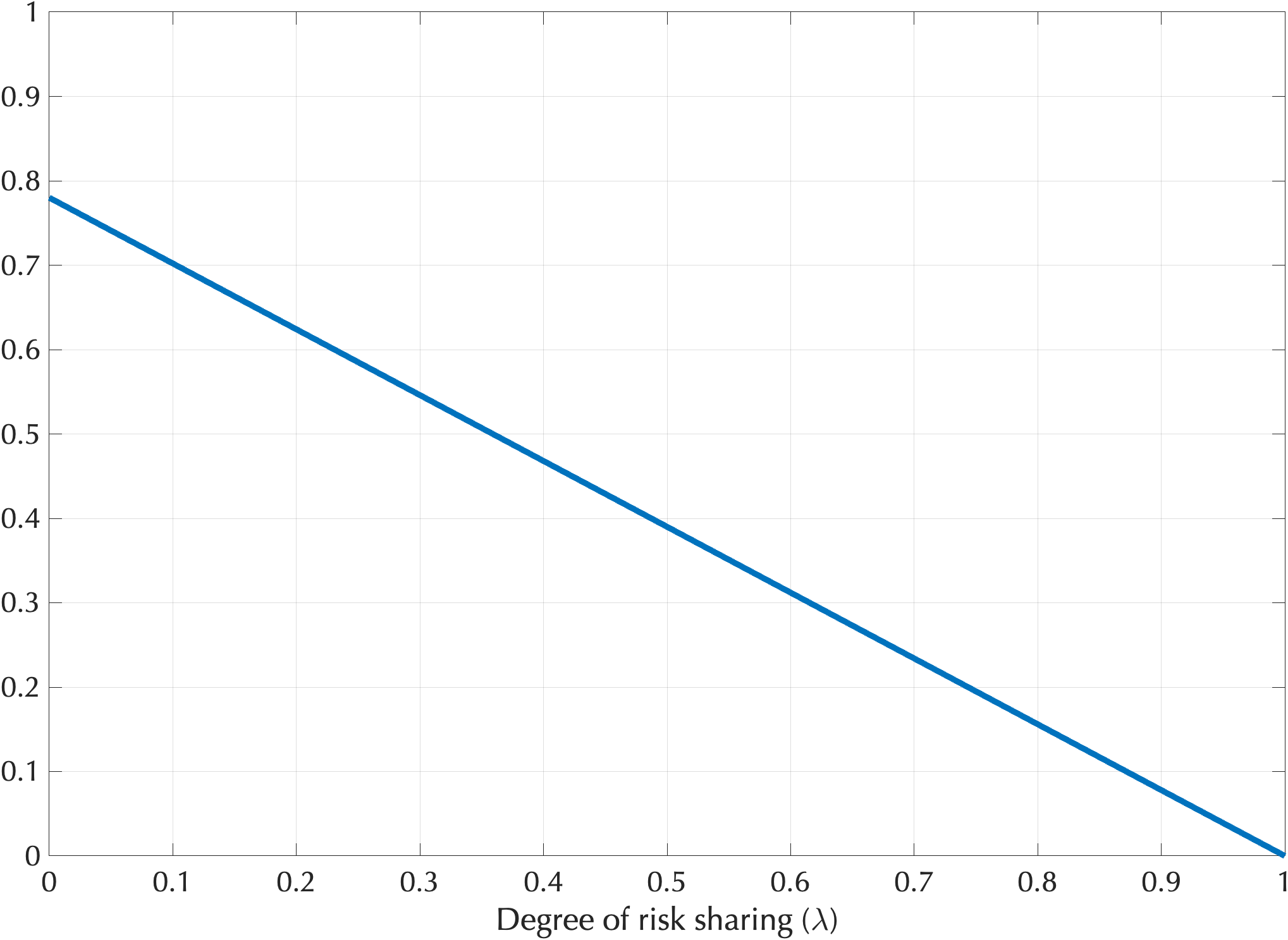}
			\caption{Share of the indirect effect}
		\end{subfigure}
		\caption{Risk sharing and multipliers} 
		\label{fig:simpleModel}
		\begin{threeparttable}
			\begin{tablenotes}[]                                                                   
				\item {\footnotesize \textit{Notes:} Panel (a) displays the relationship between the index of risk sharing on the x-axis and the open-economy relative multiplier for output, $m^Y = \frac{1}{1-\omega\bar{\mathpzc{c}} (1-\lambda)}$ and for consumption, $m^C = (1-\lambda) m^Y$, on the y-axis. The home bias parameter is set to $\omega=0.78$ and the share of consumption to $\bar{\mathpzc{c}}  \approx 1$. Panel (b) displays the share of the indirect effect in the total marginal effect of an increase in risk sharing, $\lambda$, on the consumption multiplier, $m^C$, $\mathcal{S}^{indirect} = \omega(1-\lambda)$.}
			\end{tablenotes}   
		\end{threeparttable}
	\end{figure}

	\section{Empirical Analysis} \label{sec:empirical_results}
	The previous section established that risk sharing reduces consumption volatility through two channels: (1) an \emph{indirect effect}, whereby it dampens the response of output to a given shock, and (2) a \emph{direct effect}, by which it smooths consumption for a given change in output. This section empirically quantifies these two effects using regional variation in U.S. military spending.
	
	We proceed in three steps. First, we specify our strategy for estimating the open-economy relative output multiplier, $m^Y$, and the pass-through to consumption, $\beta^C$. Second, we detail how we decompose this pass-through using the framework of ASY to identify the contributions of specific risk-sharing channels. Finally, we present our instrumental variable strategy, a key departure from the existing literature that allows for a causal interpretation of our findings. These estimates will discipline our quantitative model in Section \ref{sec:model}.

	\subsection{Direct and Indirect Effects of Risk Sharing}
	To quantify the \emph{indirect effect} or risk sharing we first estimate the open-economy relative output multiplier $m^Y$, by regressing the change in a state's nominal GDP on federal military spending shocks:\footnote{Unlike \cite{nakamura2014fiscal} and \cite{dupor2017local}, who use two-year differences, we focus on one-year differences. We can do this, since we adjust the fiscal-year military spending data to match the calendar year that the other macroeconomic variables follow. While this approach may not fully capture timing mismatches between military contract awards and actual expenditures (see \citealt{briganti2025does}), our results are robust to using a two-year difference specification, which we report in Appendix C.3.}
	\begin{align}
		\label{eq:main_regression} 
		\Delta \ln GDP_{i,t} = \alpha_i + \alpha_t + m^Y \frac{g_{i,t} - g_{i,t-1}}{gdp_{i,t-1}} + \mathbb{Z}_{i,t} +  \varepsilon_{i,t} 
	\end{align} 	
	where $GDP_{i,t}$ is aggregate nominal GDP in state $i$ at time $t$, $gdp_{i,t}$ is per-capita nominal GDP, $g_{i,t}$ denotes per-capita nominal federal military spending and $\mathbb{Z}_{i,t}$ is a vector of control variables described below. 
	
	To quantify the \emph{direct effect}, we estimate the consumption multiplier $m^c$, by replacing the outcome variable with the log change in per-capita consumption. The unsmoothed share of the income shock, $\beta^C$, is then the ratio of these two coefficients:
	\begin{align}
		\label{eq:beta_C} 
		\beta^C_t = \frac{m^c_t}{m^Y_t}
	\end{align} 	
	
	This ratio captures the pass-through from income to consumption in response to a change in external demand. We estimate standard errors for this parameter using the variance-covariance matrix of the estimates, which we estimate by stacking all regressions and saturating the model with sample indicators \citep{angrist2023methods,angrist2023instrumental}.
	
	To trace these causal effects over time, we also estimate \cite{jorda2005estimation}-style local projections for each outcome. This yields cumulative multipliers $m^Y_h$ over a horizon of $h$ years \citep{ramey2019ten}: 
	\begin{align}
		\label{eq:dynamic_regression} 
		\sum_{s=0}^h \left(\ln GDP_{i,t+s} - \ln GDP_{i,t-1}\right)  = \alpha_{h,i} + \alpha_{h,t} + m^Y_h \sum_{s=0}^h  \frac{g_{i,t+s} - g_{i,t-1}}{gdp_{i,t-1}} + \mathbb{Z}_{h,i,t} +  \varepsilon_{h,i,t} 
	\end{align}

	\subsection{Decomposition Into Risk-Sharing Channels}
	To understand the mechanisms underlying the observed degree of smoothing ($\beta^C$), we decompose the wedge between aggregate GDP and per-capita consumption. Following the seminal work by ASY, later extended by \cite{parsley2021risk} and \cite{kohler2023risk}, we use the accounting identity: 
	\begin{align} 
	GDP_{i,t} = \frac{GDP_{i,t}}{gdp_{i,t}} \cdot \frac{gdp_{i,t}}{pi_{i,t}} \cdot \frac{pi_{i,t}}{di_{i,t}} \cdot \frac{di_{i,t}}{pce_{i,t}} \cdot pce_{i,t},
	\end{align} 
	where $i$ indexes U.S. states and $t$ denotes years.  The variables represent aggregate GDP ($GDP_{i,t}$), per capita GDP ($gdp_{i,t}$), per capita personal income ($pi_{i,t}$), per capita disposable income ($di_{i,t}$), and per capita personal consumption expenditure ($pce_{i,t}$).\footnote{We convert all variables to per-capita terms using working-age population, as this demographic is more likely to respond to changing economic conditions through migration.} 
	
	Taking logs and first differences yields:
	\begin{align}
		\label{eq:decomposition}
		\begin{aligned}
			\Delta \ln GDP_{i,t} ={}& \underbrace{ \Delta \ln GDP_{i,t} - \Delta \ln gdp_{i,t} }_{\text{Migration}} + \underbrace{ \Delta \ln gdp_{i,t} - \Delta \ln pi_{i,t} }_{\text{Capital markets}} \\
			&+ \underbrace{ \Delta \ln pi_{i,t} - \Delta \ln di_{i,t} }_{\text{Fiscal transfers}} + \underbrace{ \Delta \ln di_{i,t} - \Delta \ln pce_{i,t} }_{\text{Credit markets}} + \underbrace{\Delta \ln pce_{i,t}}_{\text{Unsmoothed}}.
		\end{aligned}
	\end{align}
	This decomposition identifies four distinct risk-sharing channels that can absorb fluctuations in aggregate GDP before they translate into per-capita consumption changes, introducing a wedge between $m^Y$ and $m^c$. First, \emph{migration ($M$)}---measured as the difference between changes in aggregate and per-capita GDP---captures how worker mobility dampens local economic shocks. For example, outward migration during recessions raises per-capita GDP for remaining residents. Second, \emph{capital market smoothing ($K$)} represents the difference between changes in per-capita GDP and personal income, capturing factor income flows across regions through dividends, interest payments, commuting income, capital depreciation, and retained earnings. Third, \emph{fiscal transfers ($F$)} create a wedge between changes in personal income and disposable income, reflecting cross-regional redistribution through taxes and transfers. Fourth, credit markets allow households to smooth consumption through borrowing and lending in \emph{credit markets ($B$)}, represented by differences between changes in disposable income and consumption. Any fluctuations not absorbed by these four channels remain unsmoothed, resulting in a direct comovement between aggregate GDP and per-capita consumption.
	
	We use the different outcome variables in align (\ref{eq:decomposition}) for each risk-sharing channel. We estimate the response of each outcome variable to the military spending shock as in align (\ref{eq:main_regression}). We then calculate the normalized contribution of each channel, $\beta$, as its estimated response relative to the total output response, $m^Y$, as in (\ref{eq:beta_C}). These shares, by construction, sum to one:
	\begin{align}
	\beta^M + \beta^K + \beta^F + \beta^B + \beta^C = 1.		
	\end{align}

	\subsection{Identification: Military Spending Shocks} 
	A central contribution of our paper is the causal identification of the different risk-sharing measures. While ASY and subsequent literature have relied on unconditional variance decompositions, we estimate the response to an identified external demand shock using an instrumental variable (IV) approach. This strategy allows us to move from correlation to causation, providing a more robust foundation for our counterfactual analysis.
	
	Following \cite{nakamura2014fiscal} and \cite{dupor2017local}, we address the endogeneity of state-level military spending with a shift-share instrument. The instrument combines national changes in military spending with cross-sectional variation in states' exposure to military contracts:
	\begin{align}
		\label{eq:instr} 
		\text{instr}_{i,t} = \left(\frac{1}{2}\sum_{n=t-3}^{t-2}\frac{s_{i,n}^G}{s_{i,n}^{GDP}}\right) \cdot \frac{g_{t} - g_{t-1}}{gdp_{t-1}},
	\end{align}
	where $\frac{g_{t} - g_{t-1}}{gdp_{t-1}}$ is the national shift component (changes in per-capita military spending relative to lagged per-capita GDP), and $\frac{s_{i,n}^G}{s_{i,n}^{GDP}}=\frac{g_{i,n}/g_{n}}{gdp_{i,n}/gdp_{n}}$ captures each state's relative exposure to military spending. States with historically higher military spending per dollar of GDP experience larger shocks when national military spending changes. This identification strategy is valid if national military spending decisions reflect geopolitical considerations rather than state-specific economic conditions.
	
	To ensure instrument validity, we implement several precautions based on recent methodological advances in shift-share designs \citep{adao2019shift, goldsmith2020bartik, borusyak2022quasi}. First, we lag exposure shares by two years to reduce concerns about their potential endogeneity. Second, we include exposure shares as control variables to account for potential correlations between these shares and economic outcomes. Third, we incorporate time fixed effects to absorb aggregate temporal variation such as variation in monetary policy, federal tax policy and (lags of) national shifts in military spending. Fourth, we include lags of state-level military spending changes to isolate unexpected spending shocks.\footnote{Our results are robust to adding a greater set of control variables, including lags of dependent variables and further lags of spending changes.} Finally, we obtain unbiased standard errors by clustering our regressions at the treatment level---the state level in this case. To ensure our regressions account for the relative importance of different states for output fluctuations, we weight our regressions by the states' shares of national GDP.
	
	\subsection{Data} 
	\paragraph{Geographical Coverage and Sample Period.}
	Our analysis covers the 50 U.S. states. The sample spans 1966--2019, with 1966 marking the first year for which state-level military spending data are available.
	
	\paragraph{Data sources.}
	Data on state GDP, personal income, disposable income, and consumption are from the Bureau of Economic Analysis (BEA).\footnote{Appendix B.1 provides additional details on data sources and includes a table decomposing the relationship between GDP and personal consumption.} We focus on consumption of non-durable goods and services, as this measure aligns more closely with the standard international risk-sharing framework where purchases of goods coincide with their consumption.\footnote{As shown in Appendix C.2, durable consumption responds more strongly to shocks to output, consistent with the view that durables act as an additional saving device.} Data on state-level working-age population are from the National Cancer Institute's Surveillance, Epidemiology, and End Results Program.
	
	Military spending data are constructed from the universe of U.S. military prime contracts, aggregated to the state level. Data before 2004 are from the Department of Defense's Directorate for Information Operations and Control archives; data from 2004 onward are from USAspending.gov. We convert the data from fiscal to calendar years by assuming equal monthly spending within a fiscal year.\footnote{For instance, starting in 1977, fiscal years run from October 1st of the previous calendar year to September 30th of the current calendar year. We allocate one-quarter of that fiscal year's spending to the previous calendar year and three-quarters to the current calendar year. This adjustment primarily affects the impact estimates, where its omission attenuates the multiplier due to the mismatch between the shock and outcomes. The discrepancy fades in the medium run. Appendix C.4 presents the full comparison.} All variables are in nominal terms.
	
	\paragraph{Consumption Proxy.}
	Official state-level consumption data are available only from 1997 onward. For prior years, we construct a proxy for state-level consumption growth by allocating national consumption growth to states based on their employment growth in sectors producing non-tradable, non-durable consumption goods. This proxy tracks actual consumption dynamics closely: in the years where both series are available (1997 onward), the average correlation between our proxy and actual state-level consumption growth is 0.94, significantly outperforming retail-based alternatives.\footnote{Alternative proxies used in the literature, such as retail employment growth (0.91) or retail sales growth (0.22), show weaker correlations \citep{asdrubali1996channels, hoffmann2019channels, guren2021we}. See Appendix B.2 for construction details.}
	
	Still, we also re-ran our estimates for the output and consumption multipliers on a shorter sample starting in 1997 (see Appendices B.2 and C.7). Since there is limited identification power of the post-1997 sample, we use a two-sample two-stage least squares approach where we estimate the first stage on the full sample and the second stage on the later sample where we only use official BEA consumption data. We find that our results are not driven by measurement error in the proxy, with our estimates of $\beta^C$ being almost the same in the two samples. 
	
	\subsection{Results}
	Table~\ref{tab:main_results} presents our main empirical results. We first establish the strength of our instrumental variable approach using the weak-instrument-robust test of \cite{olea2013robust}. The effective first-stage F-statistic is 113.6 on impact, which exceeds the 5-percent worst-case bias critical value of 37.4 and indicating strong instrument relevance.
	
	\begin{table}[!ht]
		\centering
		\caption{\textsc{The Effects of Military Spending Shocks on Output and Consumption}}
		\label{tab:main_results}
		\begin{threeparttable}
			
			\sisetup{
				table-format          = 1.2,
				input-symbols         = (),
				table-space-text-post = ***,
				input-decimal-markers = .
			}
			
			\renewcommand{\arraystretch}{1}

			\newcommand{\hdrA}{\shortstack{Impact}}
			\newcommand{\hdrB}{\shortstack{Medium-Run}}
			\newlength{\blockw}
			\settowidth{\blockw}{\hdrB}
			
			\begin{tabular}{l d{2.3} d{2.3}}
				\toprule
				& \multicolumn{1}{c}{\makebox[\blockw][c]{\hdrA}}
				& \multicolumn{1}{c}{\makebox[\blockw][c]{\hdrB}} \\
				& \multicolumn{1}{c}{(1)}      & \multicolumn{1}{c}{(2)}      \\
				\midrule
				\multicolumn{3}{l}{\textit{Panel A: Multipliers}} \\
				\addlinespace
				\quad Output Multiplier ($m^Y$)         & 1.31^{**}   & 2.11^{***} \\
				& \footnotesize (0.58) & \footnotesize (0.43) \\
				\addlinespace
				\quad Consumption Multiplier ($m^C$)     & 0.62^{***}  & 0.71^{***} \\
				& \footnotesize (0.22) & \footnotesize (0.21) \\
				\midrule
				\addlinespace
				\multicolumn{3}{l}{\textit{Panel B: Decomposition of Risk-Sharing Channels}} \\
				\quad Migration ($\beta^M$)              & 0.07         & 0.19^{***} \\
				& \footnotesize (0.10) & \footnotesize (0.06) \\
				\addlinespace
				\quad Capital Markets ($\beta^K$)        & 0.12         & 0.15   \\
				& \footnotesize (0.22) & \footnotesize (0.10) \\
				\addlinespace
				\quad Fiscal Transfers ($\beta^F$)       & 0.31^{***}   & 0.22^{***} \\
				& \footnotesize (0.11) & \footnotesize (0.05) \\
				\addlinespace
				\quad Credit Markets ($\beta^B$)         & 0.02         & 0.10   \\
				& \footnotesize (0.11) & \footnotesize (0.07) \\
				\addlinespace
				\quad Unsmoothed Component ($\beta^C$)   & 0.47^{**}     & 0.34^{***} \\
				& \footnotesize (0.21) & \footnotesize (0.09) \\
				\bottomrule
			\end{tabular}
			
			\begin{tablenotes}[flushleft]
				\footnotesize
				\item \textit{Notes:} The table reports estimates of regression (\ref{eq:dynamic_regression}). The impact response sets $h=0$, the medium response sets $h=3$. Panel A reports the open-economy relative output multiplier ($m^Y$) and the open-economy relative consumption multiplier ($m^c$). Panel B reports the risk-sharing coefficients ($\beta$), which decompose the fraction of the output shock absorbed by each of the four channels, see e.g. align (\ref{eq:beta_C}). The sample period is 1967–2019 ($N=2,600$). Standard errors, clustered by state, are in parentheses. \textit{First-stage $F$-statistics:} $h=0$: 113.6; $h=3$: 226.7. Full IV coefficient tables are reported in Appendix C.1. \textsuperscript{*}p$<$0.10, \textsuperscript{**}p$<$0.05, \textsuperscript{***}p$<$0.01.
			\end{tablenotes}
		\end{threeparttable}
	\end{table}
	
	Panel A of Table \ref{tab:main_results} reports our main findings on the aggregate response. On impact (Column 1), we estimate an open-economy relative multiplier of 1.31 (0.58), implying that an increase in federal military spending amounting to 1 percent of state GDP raises state-level GDP by 1.31 percent in the year of the shock. This point estimate is comparable to the two-year multiplier reported by \cite{nakamura2014fiscal}. Critically, we find that 47 percent of this military-driven increase in output passes through to consumption (the unsmoothed component), as per-capita consumption only rises by 0.62 percent.
	
	
	Panel B decomposes how the remaining 53 percent of the increase in state GDP is absorbed through different risk-sharing mechanisms. The migration channel provides modest initial smoothing, with net in-migration offsetting 7 percent of the aggregate shock. Capital markets smooth an additional 12 percent, primarily through retained earnings or dividend payments to non-residents, though this estimate is imprecise. The federal fiscal transfer system provides the biggest smoothing at 31 percent. Finally, credit markets have a nearly negligible role, absorbing only 2 percent of the income shock, although this coefficient is not statistically significant on impact.
	
	The dynamic response changes somewhat over the medium run (column 2). The output response strengthens over time, with the cumulative multiplier reaching 2.11 three years after the shock. The overall share of unsmoothed income ($\beta^C$) falls to 0.34. This increase in smoothing is driven almost entirely by labor mobility; the migration channel's contribution more than doubles to 0.19, becoming one of the dominant smoothing mechanism in the medium term, consistent with empirical findings in \cite{foschi2025should}. The role of fiscal transfers is reduced by a third, with this decrease being picked up by increasing smoothing through credit markets. All coefficients gain statistical significance over longer horizons, driven in part by increased instrument relevance.

	\subsection{Comparison to OLS} 
	Our findings challenge the conventional view that almost all idiosyncratic state-level income fluctuations are shared across U.S. states. Figure~\ref{fig:IVvsOLS} illustrates this stark contrast against unconditional OLS estimates. 
	We estimate that a military-driven 1 percent increase in GDP causes state-level consumption to increase by 0.47 percent. In contrast, prior literature and our own OLS specification estimate that consumption only rises by 0.14 percent when GDP goes up by 1 percent (see Figure~\ref{fig:IVvsOLS}).\footnote{Using a slightly different sample, \cite{kohler2023risk} and \cite{parsley2021risk} find pass-throughs of 0.19 and 0.16 into consumption.} 
	\begin{figure}[ht!]
		\centering
		\includegraphics[width=0.7\textwidth]{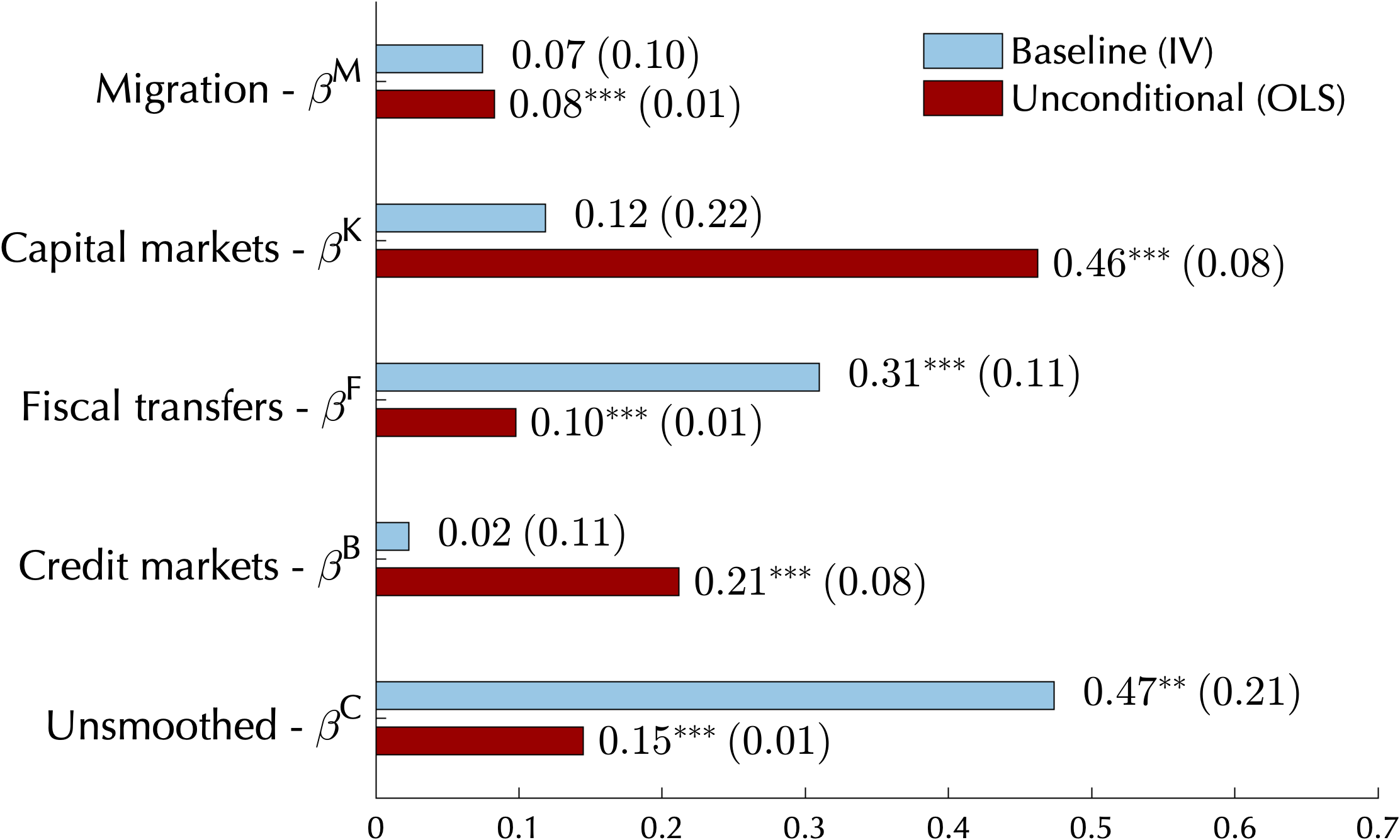}
		\caption{\textsc{Comparison of IV vs OLS risk-sharing channels}}
		\label{fig:IVvsOLS}
		
		\begin{threeparttable}    
			\begin{tablenotes}[flushleft]
				\item \footnotesize{This figure compares risk-sharing coefficients estimated using military spending shocks (IV) with those from unconditional regressions of each channel on GDP growth (OLS). For instance, the $\beta^M$ coefficient in the OLS specification is estimated from regressing $\Delta \ln GDP_{i,t} - \Delta \ln gdp_{i,t}$ on $\Delta \ln GDP_{i,t}$ and a set of state and time fixed effects. } 			                                           
			\end{tablenotes}    
		\end{threeparttable}
	\end{figure}
	
	\subsubsection{Factor Incidence and the Capital Market Channel}
	This striking discrepancy is driven strongly by the capital market channel. Unconditional variance decompositions find that capital markets are the single largest channel for smoothing state-specific shocks. Our OLS regressions confirm this, indicating they absorb 46 percent of every 1 percent increase in output. In our IV specification, however, the contribution of this channel is a statistically insignificant 12 percent. 
	
	This finding raises a puzzle: why does the capital market channel, which appears to be the linchpin of regional risk sharing in unconditional data, play such a diminished role in response to a well-identified government demand shock?
	
	The answer is that risk-sharing channels are not deep structural parameters, but depend critically on the nature of the underlying shock hitting the economy. OLS-based variance decomposition estimates are unconditional, capturing responses to \textit{all} shocks to GDP, whereas our IV estimates isolate a specific government demand shock. We next show that the unconditional dominance of the capital market channel reflects the fact that typical state-level GDP fluctuations are driven by shocks that disproportionately accrue to capital. Since capital income is far more geographically diversified than labor income, shocks that load heavily on capital income are more easily smoothed. Government spending shocks, in contrast, have a more balanced incidence on labor and capital income.
	
	\begin{table}[ht!]
		\centering
		\caption{\textsc{Elasticities of income components to GDP}}
		\label{tab:income_paid_received}
		\begin{threeparttable}
			
			\centering
			\begin{tabular}{l
					S[table-format=1.2] S[table-format=1.2]
					S[table-format=1.2] S[table-format=1.2]}
				\toprule
				& \multicolumn{2}{c}{Labor}
				& \multicolumn{2}{c}{Capital} \\
				\cmidrule(lr){2-3} \cmidrule(lr){4-5}
				& \multicolumn{1}{c}{Unconditional} & \multicolumn{1}{c}{Military spending} & \multicolumn{1}{c}{Unconditional} & \multicolumn{1}{c}{Military spending} \\ & \multicolumn{1}{c}{(OLS)} & \multicolumn{1}{c}{(IV)} & \multicolumn{1}{c}{(OLS)} & \multicolumn{1}{c}{(IV)} \\
				\midrule
				Paid     
				& 0.48$^{***}$ & 1.02$^{***}$ & 1.49$^{***}$ & 1.02$^{***}$ \\
				& {\footnotesize (0.03)} & {\footnotesize (0.27)} &  {\footnotesize (0.04)} & {\footnotesize (0.34)} \\[0.3em]
				Received 
				& 0.45$^{***}$ & 0.95$^{***}$ & 0.62$^{***}$ & 0.55$^{*}$ \\
				& {\footnotesize (0.04)} &  {\footnotesize (0.30)}  & {\footnotesize (0.18)} &  {\footnotesize (0.33)} \\		
				\bottomrule
			\end{tabular}
			
			\begin{tablenotes}[para,flushleft]
				\footnotesize
				\emph{Notes}: Elasticities of state-level labor and capital income with respect to GDP based on unconditional OLS or IV estimates conditional on military spending shocks. ``Paid'' refers to income paid in a state, ``Received'' refers to income received by residents. Sample: 1967--2019. 			
			\end{tablenotes}
			
		\end{threeparttable}
	\end{table}

	To substantiate this argument, we estimate the elasticities of state-level labor and capital income with respect to state GDP, comparing unconditional OLS estimates with our IV estimates. Table~\ref{tab:income_paid_received} displays the results. The unconditional OLS estimates reveal a strongly countercyclical labor share: a 1 percent increase in GDP is associated with only a 0.46 percent increase in labor income paid within the state, but a 1.49 percent increase in capital income. The IV estimates, however, show that the labor share is acyclical in response to a government spending shock, with the elasticities of both labor and capital income paid being approximately one.
	
	This difference in factor incidence matters critically because capital income is far more geographically diversified. The second row of Table~\ref{tab:income_paid_received} shows this by repeating the exercise for income \textit{received} by state residents.\footnote{ Received labor income and received capital income sum to personal income.} The elasticity of received labor income to GDP is nearly identical to that of paid labor income, indicating limited diversification through e.g. cross-state commuting. For capital income, however, the elasticity of received income is less than half that of paid income. This efficient diversification means that when shocks primarily affect capital income---as they appear to do unconditionally---the capital market channel provides substantial risk sharing. When shocks affect labor and capital proportionally---as government spending shocks do---this channel is naturally less important.
	
	The finding that the state-level labor share is strongly countercyclical unconditionally but acyclical in response to a demand shock raises a deeper question: What generates such a strongly countercyclical labor share in the unconditional state-level data? A canonical Cobb-Douglas production function with perfect competition predicts constant factor shares, inconsistent with our findings. Models with sticky wages or sticky prices also fail to generate countercyclical labor shares.\footnote{Sticky wage models predict constant labor shares, while sticky price models generate procyclical labor shares.} While models with overhead labor (e.g., \cite{nekarda2020cyclical,kaplan2024markups}) can generate such dynamics, this mechanism is quantitatively too weak to explain the large effect observed at the state level and does not explain why the effect is so much more pronounced at the state level than nationally.
	
	One potential explanation involves oligopolistic competition and granular shocks. In models like \cite{atkeson2008pricing}, large firms respond to positive idiosyncratic productivity or demand shocks by raising markups as they gain market power. This leads to a decline in the labor share of income. If such shocks are sufficiently geographically concentrated, they could drive local GDP while being averaged out at the national level. This class of models naturally generates volatile capital income and a countercyclical labor share, consistent with the unconditional data.
	
	An alternative, complementary explanation is measurement error. State-level labor income is derived from detailed wage data reported by employers in the Quarterly Census of Employment and Wages, while capital income is estimated indirectly and proves difficult to allocate accurately across states, particularly for firms operating in multiple states. Systematic misattribution of profits could mechanically amplify capital income volatility and inflate the apparent role of capital markets in OLS regressions.\footnote{Specifically, if the capital market channel, $\Delta \ln gdp_{i,t} - \Delta \ln pi_{i,t}$, contains measurement error, this error also enters the regressor, $\Delta \ln GDP_{i,t}$, (see align (\ref{eq:decomposition})). Regressing $\Delta \ln gdp_{i,t} - \Delta \ln pi_{i,t}$ on $\Delta \ln GDP_{i,t}$ yields a coefficient that is upward-biased due to spurious correlation from measurement error appearing on both sides. The coefficients for the other risk-sharing channels would be biased downward due to measurement error of the independent variable.} Our IV estimates are immune to this type of measurement error.\footnote{This discrepancy is consistent with critiques of the unconditional variance decomposition approach; for instance, \cite{del2002asymmetric} argues that much of the smoothing found in the prior literature may be illusory, resulting from measurement error in output rather than genuine risk sharing.} 	
	
	\subsubsection{The Fiscal Transfer System as an Automatic Stabilizer}	
	While the capital markets channel collapses in our IV estimates, the fiscal channel strengthens significantly, rising to 0.31 from 0.10 in OLS. That is, almost a third of all fluctuations in GDP are absorbed through the tax and transfer system according to the IV estimate, rather than just 10 percent. Part of this difference is mechanical because the tax and transfer system operates on per-capita personal income, which responds less to unconditional GDP fluctuations than to GDP fluctuations induced by military spending. As migration and income diversification via capital market already smooth out a large share of unconditional GDP fluctuations, there is less of a role to play for the tax and transfer system.  
	
	A more meaningful comparison is to look at coefficients that refer to a 1 percent increase in per-capita personal income rather than in GDP, i.e. to look at $\frac{\beta^F}{1-\beta^M-\beta^K}$ rather than $\beta^F$. This reveals that for our IV estimate about 0.38 (0.08) of an increase in per-capita personal income is absorbed by the fiscal transfer system. Re-doing the same calculation for the unconditional estimates reveals a lower number of 0.21 (0.04).\footnote{For our baseline military spending shocks, this means normalizing the estimated response of of fiscal transfers to a change in government spending not by the total output response (as in align (\ref{eq:beta_C})), but by the response of per-capita personal income. To get standard errors in the IV case, we proceed as explained after align (\ref{eq:beta_C}). For the OLS case, we regress changes in fiscal transfers, $\Delta \ln pi_{i,t} - \Delta \ln di_{i,t}$, on the log change in per-capita personal income, $\Delta \ln pi_{i,t}$, using $\Delta \ln GDP_{i,t}$ as an instrument. This is equivalent to calculating $\frac{\beta_{OLS}^F}{1-\beta_{OLS}^M-\beta_{OLS}^K}=0.21$.}. That is, the IV estimates still assign a substantially larger role to the tax and transfer system than the unconditional estimates, even after conditioning on the same size change in personal income. 
	
	Part of the gap can be accounted for by the documented differences in elasticities of factor income to GDP. The relevant marginal tax rate is a weighted average of the marginal tax rates on capital income and labor income. As reported by \cite{CBO2014TaxingCapitalIncome}, the marginal tax rate is lower on capital income than on labor income. Since capital income fluctuates more strongly in the unconditional data than in response to the military spending shock, the marginal tax rate is also smaller in the unconditional data. But quantitatively, this effect is small.  
	
	A second explanation lies in the fact that the OLS-based variance decomposition captures responses to all shocks to GDP, including e.g. tax policy shocks that can bias the estimates of the fiscal channel downwards. In the last 60 years, the United States has experienced several large federal tax reforms that have had an unequal impact on U.S. states because the income distribution differs across states.  \cite{zidar2019tax} exploits this fact to construct state-level tax shocks induced by changes in federal tax rates. Using those tax shocks as an instrument in a regression of fiscal transfers on log changes in personal income yields an estimate of 0.19 (0.08). This small estimate helps rationalize why the OLS estimate is smaller than the estimate using military spending as an instrument.\footnote{The estimate of 0.19 suggests that a 1 percent percent drop in personal income (induced by the state-level federal tax shock), lowers disposable income by 0.81 percent. That is, despite the increase in tax rates, the overall tax revenue (less transfers) \textit{falls} due to the strong response of personal income. The tax and transfer system therefore still smooths out some of the fall in personal income. This is consistent with \citep{zidar2019tax} who reports a fairly strong GDP response to tax shocks.} 
	
	To gauge whether our IV estimate of 0.38 is plausible, we back out the effective marginal tax rates implied by this figure and the average share of disposable to personal income (0.945 in our sample). This gives a marginal rate of 42 percent ($=1 - (1-0.38) \times 0.945$), which is somewhat higher than the estimate of about 37 percent reported in \cite{barro2011macroeconomic} for the years 1967 - 2006. \cite{barro2011macroeconomic} calculate average marginal tax rate from the NBER TAXSIM program including taxes and social security contributions, but excluding transfers. More recently, \cite{altig2020marginal} estimate a median marginal tax rate inclusive of transfers of 43 percent for 2018, close to the tax rate implied by our regression results. Our IV estimate therefore seems in line with previous estimates of the marginal tax rate.

	\section{Multi-Region Quantitative Model} 
	\label{sec:model}
	We find that close to half of state-specific income fluctuations induced by government spending shocks are smoothed through various risk-sharing mechanisms. The corresponding open-economy relative multiplier of government spending shocks is 1.31, implying that a \$1 increase in government spending raises GDP by \$1.31. In Section \ref{sec:simple_model}, we proposed a simple model that links these two measures, arguing that risk sharing reduces the size of the multiplier. The model was purposefully simple and only had a generic risk-sharing mechanism to focus on the qualitative result. Still, as we can see from Figure \ref{fig:simpleModel}, which displays the multiplier as a function of the amount of risk sharing, the simple model's \textit{quantitative} predictions are reasonably accurate as well, as it only slightly overpredicts the multiplier for the observed amount of risk sharing (1.58 rather than 1.31).  
	
	But the simple model has its limitations. First, it restricts the trade elasticity to one and therefore cannot match both the output and consumption multiplier. Second, it does not allow us to run counterfactuals that shut down one risk-sharing channel at a time. Finally, it does not speak to possibly heterogenous effects of risk sharing across U.S. states either. In this section, we correct these shortcomings and present a quantitative, multi-region model that replaces the generic risk-sharing mechanism by several risk-sharing channels that we can map more easily to the data. 
	
	The world economy consists of $N$ U.S. states that form a currency union, as well as a rest of the world (RoW). Several model ingredients are standard in this class of models, such as sticky wages, Armington trade and a production process that requires labor and intermediates. As in the simple model, we focus on external demand shocks, modeled as state-specific unpredicted changes in federal government spending, as a driving force.
	
	We complement our baseline model with four features to capture the distinct risk-sharing channels emphasized in the empirical section: First, we allow for labor mobility across states as in \cite{house2025quantifying} to model the migration channel. Second, households' stock holdings are potentially biased towards domestic firms to capture imperfect risk sharing through capital markets. Third, the fiscal transfer channel is modeled through a progressive federal income tax as in \cite{heathcote2017optimal}. Finally, households can trade non-contingent bonds subject to adjustment costs to capture the credit market channel. 
	
	We include a RoW aggregate to account for U.S. states' international trade. Since our focus is on risk sharing \textit{within} the United States in response to state-specific shocks, we treat the U.S. as a small open economy relative to the RoW in the sense that we ignore any feedback effects of U.S. shocks on the RoW. We abstract from any risk sharing between the United States and the RoW through migration, fiscal transfers, or cross-border equity \& bond holdings.

	\subsection{Households, Population and Migration}
	Each U.S. state is populated by immobile capital owners and mobile workers. The number of capital owners is fixed at $\mathbb{N}^k_i$, while the number of workers $\mathbb{N}^w_{i,t}$ can vary as they relocate across states. Total state population is
	\begin{align}
		\label{eq:defPop}
		\mathbb{N}_{i,t}=\mathbb{N}_i^{k}+\mathbb{N}_{i,t}^{w}.
	\end{align}%
	The U.S. population is constant and normalized to 1: $\sum_{i=1}^N \mathbb{N}_{i,t} = 1$. We denote aggregate variables by capital letters and per-capita values by lower-case letters, e.g. $C_{i,t}$ for aggregate consumption and $c_{i,t} = C_{i,t}/\mathbb{N}_{i,t}$ for per-capita consumption. Variables specific to workers or capital owners are expressed per individual of that type, e.g. $c^w_{i,t} = C^w_{i,t}/\mathbb{N}^w_{i,t}$.


	\subsubsection{Capital Owners} 
	Capital owners receive income from dividends, interest, and bond holdings. In state $i$, they own a fixed share $\kappa$ of local firms, which pay dividends $P_{i,t} d_{i,t}$ per owner, where $P_{i,t}$ is the price of state $i$’s final good and $d_{i,t}$ are real dividends. They also own a state-specific share $\kappa^{US}_i$ of a nationally diversified portfolio, which pays $d^{US}_t \equiv \frac{\sum_i P_{i,t} D_{i,t}}{\sum_i \mathbb{N}^k_i}$, with $D_{i,t} = \mathbb{N}^k_i d_{i,t}$. The parameter $\kappa$ governs home bias: as $\kappa \to 0$, portfolios are fully diversified. \textit{We later calibrate $\kappa$ to match the empirically observed strength of the capital market channel.} The parameters $\kappa^{US}_i$ are chosen so that net dividend flows are zero in steady state.\footnote{In particular, we impose $\kappa \bar{d}_i + \kappa^{US}_i \bar{d}^{US} = \bar{d}_i$ in the non-stochastic steady state, implying $\kappa^{US}_i = (1-\kappa) \frac{\bar{d}_i}{\bar{d}^{US}}$. That is, capital owners in high-productivity states that pay out more dividends than the national average, $\bar{d}_i > \bar{d}^{US}$, can afford more of the national portfolio, $\kappa^{US}_i > 1-\kappa$.}
	
	Capital owners also trade bonds. A bond purchased in $t-1$ with face value $b_{i,t-1}$ pays nominal interest $i_{t-1}$ in $t$. Adjusting bond position entails quadratic costs $\frac{\iota \widebar{gdp}_i}{2} \left( \frac{b_{i,t-1}}{\widebar{gdp}_i}\right)^2$, where $\widebar{gdp}_i$ is steady-state per-capita GDP and $\iota \geq 0$ scales the adjustment cost.\footnote{These adjustment costs are common in the literature on open-economy models because they induce stationarity \citep{schmitt2003closing}. A very small value of $\iota$ is enough to induce stationarity. More recently, calibrations with larger values of $\iota$ have been proposed to model segmentation in international financial markets, see e.g. \cite{maggiori2022international}.} The higher $\iota$, the more costly it is for capital owners to use bonds to smooth income shocks. We scale $\iota$ by per-capita GDP, which makes adjustment costs equally bite across U.S. states. \textit{We later calibrate $\iota$ to match the empirically observed strength of the credit market channel.} Taken together, nominal capital income is
	\begin{align}
	P_{i,t} y^k_{i,t} = \kappa P_{i,t} d_{i,t} + \kappa^{US}_i d^{US}_t + b_{i,t-1} i_{t-1} - \frac{\iota \widebar{gdp}_i}{2} \left( \frac{b_{i,t-1}}{\widebar{gdp}_i}\right)^2.	
	\end{align}
	Income is taxed according to a log-linear tax and transfer function \citep{heathcote2017optimal}: For any pre-government income $P_{i,t} y^k_{i,t}$, the disposable income is given by $\mu_{i,t} \left(P_{i,t} y^k_{i,t} \right)^{1-\tau}$, where $\mu_{i,t}$ captures the tax level (and might vary to balance the government budget), and $\tau < 1$ measures tax progressivity. If $\tau=0$, taxation is proportional, so disposable income rises one-for-one with pre-tax income and the fiscal system does not contribute to consumption smoothing. Positive values of $\tau$ indicate a progressive tax system. In general, a 1 percent increase in pre-tax income raises post-tax income by $1-\tau$ percent. The larger $\tau$, the stronger the fiscal transfer channel. \textit{We later calibrate $\tau$ to match the empirically observed strength of the fiscal transfer channel.}
	
	Capital owners allocate disposable income between consumption and savings. They choose consumption, $c^k_{i,t}$, and bond holdings, $b_{i,t}$, to maximize expected lifetime utility $\mathbb{E}_t \sum_{j=0}^\infty \beta^j \ln \left( c^k_{i,t+j}\right)$
	subject to the budget constraint 
	\begin{align}
		\label{eq:budgetConstraintK}	
		P_{i,t} c^k_{i,t} + b_{i,t} - b_{i,t-1} = \mu_{i,t} \left(P_{i,t} y^k_{i,t} \right)^{1-\tau}.
	\end{align}
	The optimal demand for bonds satisfies the following Euler align 
	\begin{align}
		\label{eq:intlEuler}
		1&= \beta\mathbb{E}_t \left\{  \frac{c_{i,t}^k}{\pi_{i,t+1} c_{i,t+1}^{k}}   \left[ 1 + (1-\tau) \mu_{i,t+1} \left( P_{i,t+1} y^k_{i,t+1} \right)^{-\tau} \left( i_t - \iota \frac{b_{i,t}}{\widebar{gdp}_i} \right)\right]\right\} 
	\end{align}
	with $\pi_{i,t+1} = \frac{P_{i,t+1}}{P_{i,t}}$ denoting inflation. Without taxation ($\tau=0$, $\mu_{i,t}=0$) and adjustment costs ($\iota=0$), this collapses to the standard Euler align. 
	
	
	\subsubsection{Workers}
	Workers are mobile and earn only labor income, given by $P_{i,t} y^w_{i,t} = W_{i,t} l_{i,t}$. Their income is taxed through the same log-linear tax and transfer function as for capital owners. Assuming workers are hand-to-mouth, consumption equals disposable income:
	\begin{align}
		P_{i,t} c_{i,t}^{w}=\mu_{i,t} \left( P_{i,t} y^w_{i,t} \right)^{1-\tau}.
		\label{eq:budget_constraint_worker}
	\end{align}%
	Labor supply is inelastic, so the only choice workers face is where to work. Before turning to migration, we describe how effective labor $l_{i,t}$ is determined.
	
	\paragraph{Labor Supply per Worker} 
	Workers supply fixed amounts of labor, normalized to 1, in their state of residence. As in \cite{erceg2000optimal}, workers are randomly assigned a job $\iota \in \left[ 0,1\right] $ that pays a nominal wage $W_{i,t}(\iota)$. For each job, there is a labor union with market power that acts in the interest of its members in setting a wage rate. 
	
	Job-specific labor $l_{i,t}\left( \iota \right) $ is employed by competitive labor-aggregating firms who sell aggregate effective labor to goods-producing firms at wage $W_t$. Labor-aggregating firms choose a combination of jobs $l_{i,t}\left( \iota \right)$ to maximize their profits $\left\{W_{i,t} l_{i,t}-\int_{0}^{1}W_{i,t}(\iota )l_{i,t}(\iota )d\iota \right\}$
	subject to the definition of effective labor $l_{i,t}$:
	\begin{align}
		l_{i,t}=\zeta + \left(\int_{0}^{1} \left(l_t(\iota) - \zeta\right)^\frac{\psi_w-1}{\psi_w}d\iota\right)^\frac{\psi_w}{\psi_w-1}
		\label{labor combination},
	\end{align}
	where $\psi_w>1$ is the elasticity of substitution across jobs.\footnote{The parameter $\zeta>0$ in this aggregator ensures a positive equilibrium wage in the presence of inelastic labor supply \citep{house2025quantifying}.} This maximization problem gives rise to a labor demand curve for each job. Each period, labor unions can reset wages for individual jobs, $W_{i,t}\left( \iota \right) $, with probability $1-\theta_w$. If they can adjust the wage, labor unions set it to maximize the expected net present value of total labor income for their workers taking the demand curve for each job as given. Solving this optimization problem yields the following labor supply curve (see Appendix D.2):
	\begin{align}
		\label{eq:wagePC}
		\tilde{\pi} _{i,t}^{w}=\frac{\left( 1-\theta _{w}\beta \right) \left( 1-\theta
			_{w}\right) }{\theta _{w}}\tilde{l}_{i,t} +\beta \mathbb{E}_{t}\left[ \tilde{\pi}_{i,t+1}^{w}\right] ,
	\end{align}%
	where a tilde denotes log deviations from the non-stochastic steady state and $\pi^w_{i,t}$ is wage inflation at time $t$.

	\paragraph{Migration}
	At the start of each period, workers decide whether to migrate or remain in their
	current state. Migration takes place at the beginning of each period and
	migrants immediately work and consume in their new location.
	
	A worker moving from state $i$ to state $j$ faces a migration cost $\upsilon_j^i$ (with $\upsilon_i^i = 0$). In addition, workers draw idiosyncratic destination-specific shocks $\epsilon_{j,t}$.\footnote{Each worker draws their own $\epsilon_{j,t}$, but we suppress the individual index for notational ease.} Let $v_{i,t}(\epsilon_t)$ denote the value of living in state $i$ at time $t$ given the vector of shocks $\epsilon_t = [\epsilon_{1,t}, \ldots, \epsilon_{N,t}]$ and aggregate conditions. Then
	\begin{align}
		v_{i,t}(\epsilon_t) = \max_j \left\{ \ln c^w_{j,t} + \frac{1}{\gamma}\epsilon_{j,t} - \upsilon_j^i + \beta \mathbb{E}_t[V_{j,t+1}] \right\}.
	\end{align}
	Flow utility is $\ln c^w_{i,t}$. The ex-ante value $V_{i,t}$ averages over the idiosyncratic shocks and represents the expected utility of workers in state $i$ at the start of period $t$. The parameter $\gamma $ governs how strongly
	idiosyncratic location shocks affect migration decisions.
	
	Following \cite{artucc2010trade}, we assume shocks are i.i.d. across workers and time, drawn from a Type-I extreme value distribution with mean zero. Under this assumption,
	\begin{align}
		V_{i,t}=\frac{1}{\gamma }\ln \left\{ \sum_{j}\exp \left\{ \gamma \left(
		\ln c_{j,t}^{w} -\upsilon _{j}^{i}+\beta \mathbb{E}%
		_{t}\left( V_{j,t+1}\right) \right) \right\} \right\} .
		\label{eq:utilityWorker}
	\end{align}%
	Migration decisions depend on this average utility. Let $n^i_{j,t}$ be the fraction of workers moving from $i$ to $j$. Then, the number of workers living in state $i$ evolves as
	\begin{align}
		\mathbb{N}_{i,t}^{w}=\sum_{j}n_{i,t}^{j}\mathbb{N}_{j,t-1}^{w}.
	\end{align}
	with
	\begin{align}
		n_{j,t}^{i}=\frac{\exp \left\{ \gamma \left( \ln
			c_{j,t}^{w}  -\upsilon _{j}^{i}+\beta \mathbb{E}_{t}\left(
			V_{j,t+1}\right) \right) \right\} }{\sum_{k}\exp \left\{ \gamma \left(
			\ln c_{k,t}^{w}  -\upsilon _{k}^{i}+\beta \mathbb{E}%
			_{t}\left( V_{k,t+1}\right) \right) \right\} }.  \label{eq:migration_rate}
	\end{align}%
	States with higher expected utility attract more workers. Migration responses depend on two key parameters: the cost matrix $\mathbf{\upsilon}$, which governs steady-state migration flows, and $\gamma$, the inverse dispersion of idiosyncratic shocks. A higher $\gamma$ implies a lower shock variance, so migration responds more strongly to consumption differences. \textit{Given $\mathbf{\upsilon}$, we calibrate $\gamma$ to match the empirically observed strength of the migration channel.}

	\subsection{Firms, Production and Trade} 
	\label{sec:model_production}
	Production takes place in a two-stage process. A first set of competitive firms produce goods and sell these to other firms located either domestically or abroad. Second-stage producers, in turn, combine the first-stage goods with imported goods to produce a composite good that is used for consumption and as intermediates. 
	
	\subsubsection{First Stage} First-stage firms use capital (fixed to 1), hire labor, $L_{i,t}= \mathbb{N}^w_{i,t} l_{i,t}$, and buy intermediates, $X_{i,t}$, to produce goods according to
	\begin{align}
		\label{eq:prodInterm} 
		Q_{i,t} = A_i L_{i,t}^{(1-\alpha)(1-\chi)} X_{i,t}^\chi,
	\end{align}
	where $A_i$ is a scaling factor, and $0\leq \alpha \leq 1$ and $0 \leq \chi \leq 1$ determine the curvature of the production function. Producers own the capital stock and pay dividends to households that consist of sales, $P_{i,t}^{Q} Q_{i,t}$, less labor costs, $W_{i,t} L_{i,t}$, and costs for intermediates, $P_{i,t} X_{i,t}$: $P_{i,t} D_{i,t} =  P_{i,t}^{Q} Q_{i,t} - W_{i,t} L_{i,t} - P_{i,t} X_{i,t}$.
	Each period $t$, first-stage producers choose dividends, $D_{i,t}$, labor, $L_{i,t}$, and intermediates, $X_{i,t}$, to maximize the expected discounted sum of their dividends, $\mathbb{E}_t \sum_{s=0}^\infty  \beta^{t+s} D_{i,t+s} $ subject to the definition of dividends and (\ref{eq:prodInterm}).\footnote{Technically, firms apply a stochastic discount factor to discount their expected future dividends. The log-linearized equilibrium conditions would not be affected by this and we therefore omit it for brevity.} 
	Optimal labor demand and demand for intermediates implies
	\begin{align}
		\label{eq:optimalLD}
		W_{i,t} L_{i,t} &= (1-\alpha)(1-\chi) P^Q_{i,t} Q_{i,t} \\
		P_{i,t} X_{i,t} &= \chi P^Q_{i,t} Q_{i,t}
	\end{align}
	
	\subsubsection{Second Stage} 
	A second set of producers combine domestically produced and imported first-stage goods to produce a composite, second-stage good, $Z_{i,t}$. These producers choose inputs to maximize profits, $
	\left\{ P_{i,t} Z_{i,t} - \sum_{j=1}^{N+1} P_{j,t}^{Q} Q^j_{i,t}\right\}$, 
	subject to the CES\ production function 
	\begin{align}
		\label{eq:productionFinal} 
		Z_{i,t} = \left( \sum_{j=1}^{N+1} \left( \omega^j_i \right)^{\frac{1}{\psi}} \left( Q^j_{i,t} \right)^{\frac{\psi-1}{\psi}} \right)^{\frac{\psi}{\psi-1}}.  
	\end{align}%
	Here, the parameter $\psi$ describes the elasticity of substitution between first-stage goods, $\omega^j_i$ describes the preference weight by state $i$ firms for goods imported from $j$ (with $\sum_j \omega^j_i = 1$), and $Q^j_{i,t}$ is the quantity of first-stage goods imported from $j$ by state $i$. Notice that state $i$ also imports from the rest of the world, indexed by $N+1$. The price of the RoW good in U.S. dollars is given by $P^Q_{N+1,t}$. The optimal expenditure share on goods imported from $j$ is given by
	\begin{align}
		\label{eq:optimalDemandIntermediates}
		s^j_{i,t}:=  \frac{P^Q_{j,t} Q^j_{i,t}}{P_{i,t} Z_{i,t}} &= \omega^j_i  \left( \frac{P^Q_{j,t}}{P_{i,t}} \right)^{1-\psi}. 
	\end{align}
	We assume a symmetric structure of composite-good producers in the RoW, giving rise to a demand for international exports for each U.S. state analagous to (\ref{eq:optimalDemandIntermediates}).

	\subsection{Monetary and Fiscal Policy}
	\subsubsection{Monetary Policy}
	\label{sec:government} Monetary policy in the United States is set by the Federal Reserve.
	The Federal Reserve\ follows a \textquotedblleft Taylor rule\textquotedblright\ that
	targets GDP-weighted averages of GDP fluctuations and inflation throughout the
	United States:
	\begin{align}
		i_t=\phi i_{t-1}+(1-\phi )\left[ \bar{r}+\phi_{GDP}\sum_{i=1}^{N}\widebar{GDP}_i \cdot \widetilde{GDP}_{i,t}+\phi _{\pi }\sum_{i=1}^{N}\widebar{GDP}_i \cdot \pi _{i,t}\right] .
		\label{eq:Mankiw_Rule}
	\end{align}%
	The parameters $\phi $, $\phi_{GDP}$ and $\phi _{\pi } $ govern interest rate persistence, the
	interest rate reaction to fluctuations in GDP and the reaction to inflation, respectively.

	\subsubsection{Fiscal Policy} 
	The federal government spends $P^Q_{i} G_{i,t}$ on first-stage goods produced in state $i$. We assume that per-capita spending $g_{i,t}$ follows an AR(1) with persistence $\rho$:
	\begin{align}
		\label{eq:shock}
		\ln g_{i,t} = \rho \ln g_{i,t-1} + \epsilon_{i,t}.
	\end{align}
	The government's budget constraint is 
	\begin{align}
	\sum_i P^Q_{i,t}G_{i,t} = \sum_i \left(P_{i,t} \left(Y^k_{i,t} + Y^w_{i,t} \right) -  \mu_{i,t} \left[ \mathbb{N}^k_i \left(P_{i,t} y^k_{i,t} \right)^{1-\tau} +  \mathbb{N}^w_{i,t} \left(P_{i,t} y^w_{i,t} \right)^{1-\tau} \right] \right).	
	\end{align}
	To ensure a balanced budget at all times, the government adjusts the level of taxation, $\mu_{i,t}$, across all states. We impose that any changes in nominal government spending outside the steady state are financed by adjusting $\mu_{i,t}$ proportionally across states, i.e. $\tilde{\mu}_{i,t} = \tilde{\mu}_{j,t}$.
	
	\subsection{Market Clearing} 
	Market clearing of the first-stage good requires that the total production of goods by state $j$, $Q_{j,t}$,equals total demand, which consists of demand by second-stage producers, $\sum_{i=1}^N Q^j_{i,t}$, demand by the RoW, $Q^j_{N+1,t}$, and demand by the government, $G_{j,t}$. 
	\begin{align}
		\label{eq:mcIntermediates}	
		Q_{j,t}= G_{j,t} + \left(\sum_{i=1}^N Q^j_{i,t} \right) + Q^j_{N+1,t}
	\end{align}%
	The second-stage, composite good is used for consumption or as intermediate:
	\begin{align}
	Z_{i,t} = C_{i,t} + X_{i,t}.	
	\end{align}
	Final consumption equals consumption by capital owners and by workers: 
	\begin{align}
		\label{eq:defC}	
		C_{i,t}=  C_{i,t}^{k} + C_{i,t}^w.
	\end{align}
	Bond market clearing requires 
	\begin{align}
	\sum_{i=1}^N B_{i,t} = 0.	
	\end{align}
	Under the assumption of no risk sharing between the United States and the RoW, the total value of RoW imports from all U.S. states needs to equal the total value of RoW exports towards U.S. states:
	\begin{align}
	\sum_{j=1}^N P^Q_{j,t} Q^{j}_{N+1,t} = P^Q_{N+1,t} \sum_{i=1}^N  Q^{N+1}_{i,t}.	
	\end{align}

	\section{Model Solution and Estimation} \label{sec:model_solution} 
	
	We solve the model using a first-order approximation around a zero inflation steady state. We calibrate the model to U.S. states. The model is expressed at a quarterly frequency. 
	
	We partition the parameters into a set of calibrated parameters and a set of estimated parameters. Parameters that have commonly accepted values used in the international business cycle literature or parameters that have direct analogues in the data (e.g., trade shares, migration shares, etc.) are calibrated accordingly. Taking the calibrated parameters as given, we estimate the remaining four parameters pertaining to the risk-sharing channels---the migration sensitivity, $\gamma$, the home bias in capital markets, $\kappa$, the parameter of tax progressivity, $\tau$, and the adjustment costs on bond holdings, $\iota$---as well as the Armington elasticity $\psi$. 
	
	\subsection{Calibration} 
	Table \ref{tab:calibration} lists the calibrated parameter values for our baseline specification. While some parameters are assumed to be the same across states, our model captures states' variation in size, their exposure to trade and migration flows, and their exposure to military spending shocks.  
	
	\paragraph{Households, Population and Migration} From the Euler align, we obtain a relation between the discount factor $\beta$, the steady-state interest rate, the tax progressivity parameter $\tau$, and the share of aggregate military spending in GDP, $\bar{\mathpzc{g}}^{US}:=\bar{g}^{US}/\widebar{gdp}^{US}$: $\beta = \left( \bar{i} (1-\tau) (1-\mathpzc{g}^{US})+1\right)^{-1}$. We set the annual interest rate to 4 percent and solve for $\beta$ given values of $\bar{\mathpzc{g}}^{US}$ and $\tau$. The share of workers in the total population is $1-\alpha$, implying equal steady-state consumption for workers and capital owners. Migration costs $\upsilon_j^i$ affect the system only through steady-state bilateral migration shares $\bar{n}^i_j$. Rather than solving for the cost parameters directly, we condition on the observed bilateral migration matrix constructed from IRS data \citep{house2025quantifying}. Finally, we set the wage rigidity parameter to $\theta_w=0.84$, in line with the estimates of \citet{grigsby2021aggregate}.
	
	\paragraph{Firms, Production and Trade} We set $\alpha = 0.38$ to match a labor income share of 0.62 \citep{karabarbounis2024perspectives}. The parameter $\chi$ governs the ratio of gross output to value added; we calibrate it to $\chi = 0.438$, which corresponds to an output-to-value-added ratio of 1.78 \citep{bea2025}. States’ productivity, $A_i$, enters the system of log-linearized aligns through its effect on steady-state GDP. We directly condition on observed GDP levels from the BEA. Likewise, we condition on observed bilateral expenditure shares on imported goods, $\bar{s}^i_j$, rather than solving for preference weights $\omega^i_j$. For this, we use the dataset of \citet{rodriguez2025trade}, which provides estimates of trade in both goods and services across U.S. states and with the rest of the world. 
	Including services is important because they constitute a large share of output and are less tradable than goods. The resulting dataset delivers a bilateral trade matrix covering both interstate and international linkages. On average, trade accounts for 22 percent of state output, with one-third of this involving international partners. This trade intensity is comparable to medium-sized European economies such as Czechia or Denmark.\footnote{In 2002, the import share of gross output was 21 percent for Czechia and 24 percent for Denmark (WIOT). The average import share is below the 31 percent reported by \citet{nakamura2014fiscal}, who measure trade relative to GDP rather than gross output.} 
	
	\paragraph{Fiscal and Monetary Policy} \noindent We set the steady-state ratio of federal government purchases to GDP to its observed value at the national level. The monetary policy rule (\ref{eq:Mankiw_Rule}) is parameterized with $\phi=0.75$, $\phi_Y=0.5$, and $\phi_\pi=1.5$, consistent with \citet{clarida1999science}. The share of federal military spending in state GDP is taken from \citet{nakamura2014fiscal}, who also estimate a quarterly persistence parameter for military spending shocks of $\rho=0.93$.
	
	\begin{table}     
		\centering         
		\resizebox{1\textwidth}{!}{
			\begin{threeparttable}
				\caption[Calibration]{\textsc{Calibration}}
				\label{tab:calibration}
				\begin{tabular}{l l d{1.3} l}
					\toprule  
					Parameter &    & \multicolumn{1}{c}{\text{Value}} & Source / Target\\ \midrule 	
					
					\multicolumn{3}{l}{\textbf{Households, Population and Migration}} \\	
					Steady-state interest rate & $\bar{i}$ & 0.01 & Annual interest rate of 4 percent \\ 
					Population & $\bar{\mathbb{N}}_i$ & \multicolumn{1}{c}{\text{st.sp.}}  & Tax returns \citep{house2025quantifying}   \\
					Migration costs & $\upsilon^i_j$ & \multicolumn{1}{c}{\text{st.sp.}}    & Bilateral migration shares  \citep{house2025quantifying} \\
					Wage stickiness & $\theta_w$ & 0.84 & Wage duration of 1.5 years \citep{grigsby2021aggregate} \\ [.3cm]
					
					\textbf{Firms, Production and Trade} \\
					Labor share in value added & $1-\alpha$ & 0.62 & Labor income share \citep{karabarbounis2024perspectives} \\
					Share intermediates in output & $\chi$ & 0.44 & Output to value added `97-`19, \citep{bea2025}\\
					State size & $A_i$ & \multicolumn{1}{c}{\text{st.sp.}}    & State GDP `66-`19, \citep{bea2025b} \\
					Preference weights on goods & $\omega^i_j$ & \multicolumn{1}{c}{\text{st.sp.}}  & Expenditure shares by goods' origin, $\bar{s}^j_i$ \\
					& & &  \citep{rodriguez2025trade}	\\ 			
					
					\textbf{Monetary and Fiscal Policy} \\
					MP rule persistence & $\phi$ & 0.75 & \cite{clarida2000monetary} \\
					MP rule GDP coefficient & $\phi_{GDP}$ & 0.5 & \cite{clarida2000monetary} \\
					MP rule inflation coefficient & $\phi_\pi$ & 1.5 & \cite{clarida2000monetary} \\
					Military spending over GDP & $\bar{\mathpzc{g}}_i$ & \multicolumn{1}{c}{\text{st.sp.}}   & \cite{nakamura2014fiscal} \\ [.3cm]
					
					\textbf{Shock process} \\
					Shock persistence & $\rho$ & 0.93 & \cite{nakamura2014fiscal} \\ 
					\bottomrule  
				\end{tabular}
				\begin{tablenotes}[flushleft]                               
					\item {\footnotesize \textit{Notes:} Values marked with st.sp. are state specific or state-pair specific. }
				\end{tablenotes}                                                                                                                                                                           
			\end{threeparttable}
		}
	\end{table}

	\subsection{Estimation}
	Given the set of calibrated parameters, we employ a simulated method of moments to estimate the risk-sharing parameters—the migration elasticity ($\gamma$), the equity home bias ($\kappa$), the tax progressivity ($\tau$), the adjustment cost on bond holdings ($\iota$)—and the Armington elasticity ($\psi$). Our procedure proceeds in three steps. First, given an initial guess for $[\gamma, \kappa, \tau, \iota, \psi]$, we feed the observed changes in government spending for each state into the model. To this end, we quarterly interpolate the annual government spending data and log-linearly detrend it (See Appendix D.5.2). The model then generates time series of all endogenous variables at the state level, expressed in log deviations from the non-stochastic steady state. Second, we re-run regressions along the lines of (\ref{eq:main_regression}) on the simulated data. Since our model is only accurate up to a first order, we follow \cite{nakamura2014fiscal} and approximate the estimation align (\ref{eq:main_regression}) up to first order.\footnote{That is, the regressand is $\widetilde{GDP}_{i,t} - \widetilde{GDP}_{i,t-1}$ and the regressor is $\frac{\bar{P}^Q_i \bar{G}_i}{\widebar{GDP}_i} \left(\tilde{P}^Q_{i,t} + \tilde{g}_{i,t} - \tilde{P}^Q_{i,t-1} - \tilde{g}_{i,t-1}\right)$.} As in the empirical section, we replace the outcome variable by the various risk-sharing channels. Third, we compare our estimates of the risk-sharing betas of the four channels: $\beta^M$, $\beta^K$, $\beta^F$ and $\beta^B$, as well as the output multiplier, $m^Y$, to those found in the empirical data (see Table \ref{tab:main_results}), and adjust our initial guess until the two sets of parameter values match. 
	
	The four risk-sharing parameters closely map to the risk-sharing $\beta$'s. The migration elasticity $\gamma$ governs the responsiveness of migration to changes in labor income, while the equity home bias $\kappa$, the tax progressivity $\tau$, and the bond adjustment cost $\iota$ shape the capital, fiscal, and bond market channels, respectively. The Armington elasticity $\psi$ is closely linked to the output multiplier: a higher $\psi$ strengthens expenditure switching in response to an increase in military spending, raising imports, weakening local general equilibrium effects, and thereby lowering the multiplier.
	
	To implement this procedure, we define the model counterparts of the income concepts in the data: \emph{Aggregate nominal GDP} is 
	\begin{align}
	GDP_{i,t} = P_{i,t}^{Q}  Q_{i,t} - P_{i,t} X_{i,t}.	
	\end{align}
	\emph{Per-capita nominal GDP} is 
	\begin{align}
	gdp_{i,t} =GDP_{i,t} / \mathbb{N}_{i,t}.	
	\end{align}
	\emph{Per-capita personal income} consists of labor income of workers plus dividends received by capital owners as well as net interest:
	\begin{align}
	pi_{i,t} = \frac{1}{\mathbb{N}_{i,t}} \left[ W_{i,t} L_{i,t} + \kappa P_{i,t} D_{i,t} + \kappa^{US} D^{US}_t + B_{i,t-1} i_{t-1} \right] .	
	\end{align}
	It differs from GDP by including dividends and interest income earned abroad. Hence, the capital market channel reflects movements in differences in capital market returns. 
	\emph{Per-capita disposable income} is
	\begin{align}
	di_{i,t} = \mu_t \left( pi_{i,t} \right)^{1-\tau}	
	\end{align}
	And \emph{per-capita nominal consumption} is given by 
	\begin{align}
	pce_{i,t} = P_{i,t} c_{i,t}.	
	\end{align}

	\paragraph{Estimated Parameters} 
	Table \ref{tab:estimation} displays the estimated parameters together with the targeted moments. Since we are exactly identified, our model perfectly matches the empirical moments. Our estimate of the Armington elasticity is $\psi = 1.17$, which is a bit higher than what is typically estimated from international data \citep{boehm2023long}, but in line with the notion that trade across U.S. states encounters lower barriers and might therefore be more responsive to relative price changes. The estimated migration elasticity is $\gamma=5.78$, which is substantially larger than what is estimated by \cite{house2025quantifying} for a set of European countries, in line with the idea that workers are more mobile across U.S. states than across European countries. The equity home bias is estimated to be $\kappa=0.69$, implying that a \$1 increase in capital income generated in state $i$ raises received capital income in that state by \$0.69. This number is in line with the shares of the different types of capital income. Over our sample period, about 40 percent of capital income is likely to be generated locally (proprietors' income including taxes, rental income), whereas the remaining 60 percent are easier to diversify (corporate profits including taxes and net interest). Our estimate of $\kappa$ suggests that about half of the diversifiable capital income is actually diversified.\footnote{Over our sample period, about a quarter of capital income stems from proprietors' income, generally income generated by self-employed (e.g. restaurant owners) that, in practice, is linked to the household's home state. Another 5 percent stems from rental income, which might also be biased towards a household's home state. About 45 percent of capital income consists of corporate profits and interest rate income, where geographical diversification is easier. And another quarter relates to taxes on production and imports less subsidies. Distributing those taxes proportionately to proprietors' income and corporate profits suggests that about 40 percent of capital income is accounted for by proprietors' income including taxes and rental income.} Our estimate of tax progressivity is $\tau=0.38$, which translates into a marginal tax rate of about 40 percent, in line with the evidence from \cite{altig2020marginal} and somewhat above the 34 percent estimate found by \cite{heathcote2017optimal}. Bond adjustment costs for capital owners are estimated to be somewhat large to match the small credit market channel. The estimate of $\iota=0.05$ implies that a bond-to-GDP position that deviates 10 percentage points from steady state entails costs of about 0.2 basis points of GDP.  

	\begin{table}     
		\centering             
		\begin{threeparttable}
			\caption[Estimation]{\textsc{Estimation}}
			\label{tab:estimation}
			\begin{tabular}{l d{2.3} d{2.3} c l d{2.3}}
				\toprule  
				\multicolumn{1}{l}{\text{Moment}} & \multicolumn{1}{c}{\text{Data}} & \multicolumn{1}{c}{\text{Model}}  && \multicolumn{1}{l}{\text{Parameter}} & \multicolumn{1}{l}{\text{Value}} \\ \midrule 
				Output multiplier ($m^Y$) & 1.31 & 1.32 && Trade elasticity ($\psi$) & 1.17 \\
				Migration channel ($\beta^M$) & 0.07 & 0.07 && Migration elasticity ($\gamma$) & 5.78 \\
				Capital market channel ($\beta^K$) & 0.12 & 0.12 && Equity home bias ($\kappa$) & 0.69 \\
				Fiscal transfer channel ($\beta^F$) & 0.31 & 0.31 && Tax progressivity ($\tau$) & 0.38 \\
				Credit market channel ($\beta^C$) & 0.02 & 0.02 &&  Bond adj. cost ($\iota$) & 0.05 \\
				\bottomrule 
			\end{tabular}  
			\begin{tablenotes}[flushleft]                               
				\item {\footnotesize \textit{Notes:} The table displays the targeted moments and estimated parameters. The targeted moments correspond to those reported in Table~\ref{tab:main_results}. Column 2 reports the values of those moments as estimated from the actual data. Column 3 reports the values as estimated from the simulated data.}
			\end{tablenotes}   		     
		\end{threeparttable}
	\end{table}

	\section{Quantitative Results} \label{sec:counterfactuals} 
	This section uses our estimated structural model to quantify the benefits of risk sharing. We conduct a series of counterfactual experiments to assess how interstate risk-sharing mechanisms attenuate regional business cycles in the United States. In each experiment, we feed the same sequence of federal military spending shocks into the model and examine how the output and consumption multipliers change.

	\subsection{Multipliers Without Any Risk Sharing}
	Our main counterfactual shuts down all risk-sharing channels by setting the migration elasticity to zero, $\gamma \rightarrow 0$, imposing complete equity home bias $\kappa = 1$, implementing a linear tax system, $\tau=0$, and making the adjustment costs on bonds go towards infinity, $\iota \rightarrow \infty$. We then feed the observed sequence of military spending shocks into this counterfactual economy and re-estimate the output and consumption multipliers from the simulated data.
	
	Figure \ref{fig:main_counterfactual} displays the main result of this exercise. In the baseline calibration, the output multiplier is 1.32, as observed in the data. In the counterfactual economy without any risk sharing, the output multiplier rises to 2.55, a 93 percent increase. This demonstrates that the risk-sharing mechanisms currently operating across U.S. states substantially attenuate regional business cycle fluctuations.

	\begin{figure}[h!]
		\centering
		\includegraphics[width=.75\linewidth]{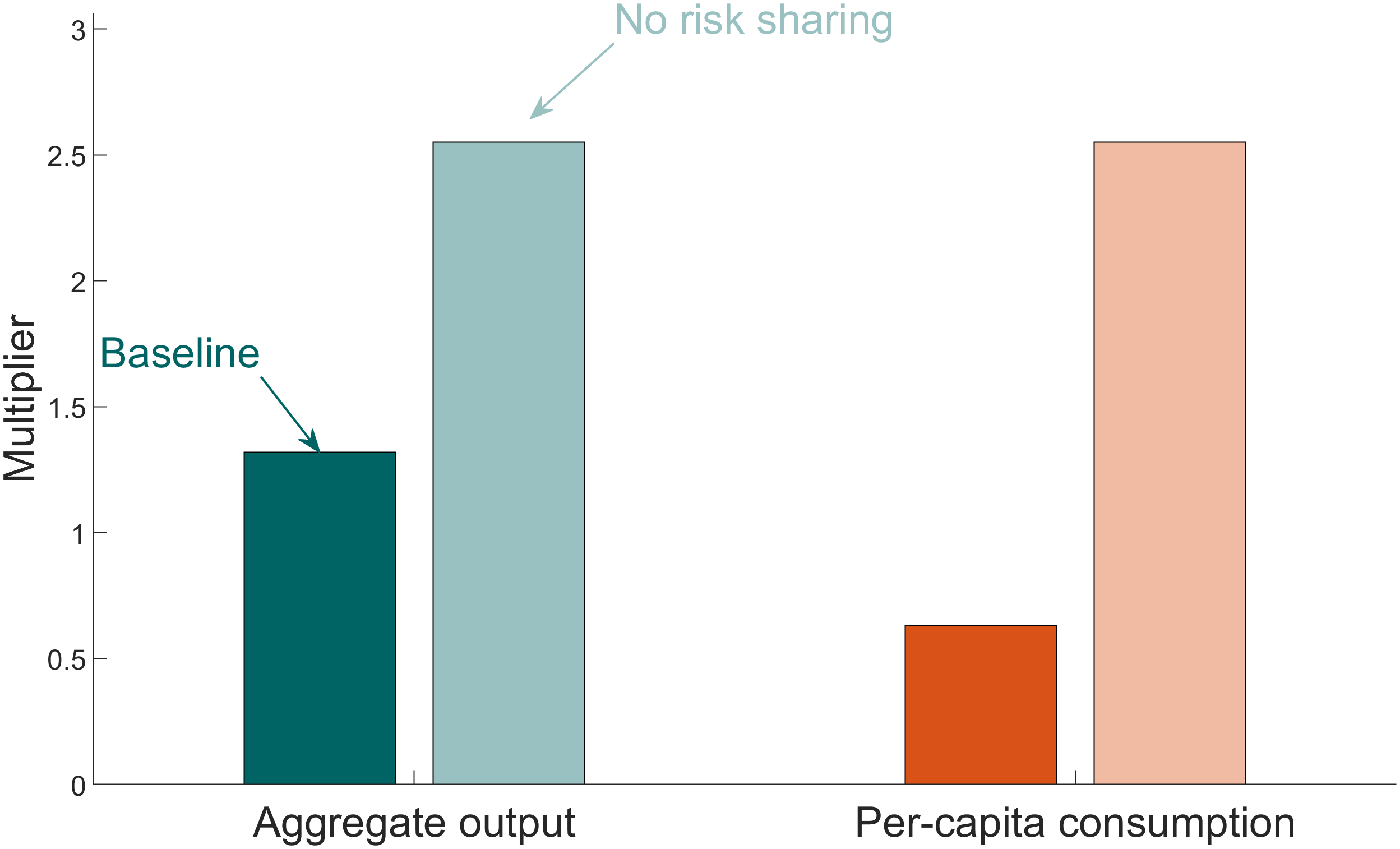}
		\caption{\textsc{Multipliers with and without risk sharing}}
		\label{fig:main_counterfactual}	
		\begin{threeparttable}
			\begin{tablenotes}[]                                                                   
				\item {\footnotesize \textit{Notes:} The figure displays multipliers for aggregate nominal GDP and per-capita nominal consumption. The estimates are based on model data, simulated from the benchmark calibration and a counterfactual calibration that shuts down all risk-sharing channels ($\gamma \rightarrow 0$, $\kappa=1$, $\tau=0$ and $\iota \rightarrow \infty$).}
			\end{tablenotes}   
		\end{threeparttable}
	\end{figure}
	
	This amplification of output volatility translates into an even larger increase in consumption volatility. In the baseline, the consumption multiplier is 0.64, consistent with the empirical finding that 47 percent of the income shock passes through to consumption. If we were to only consider the direct effect of eliminating risk sharing---that is, raising the pass-through from 47 percent to 100 percent while holding the output multiplier fixed at its baseline value of 1.32---the consumption multiplier would rise to 1.32.
	
	However, the full impact on consumption volatility is far greater because the absence of risk sharing also amplifies the output multiplier itself. When this indirect general equilibrium effect is taken into account, the consumption multiplier rises from 0.64 in the baseline to 2.55 in the counterfactual. 
	Based on align (\ref{eq:response_C}), we can calculate the share of the indirect effect as the change in the output multiplier (as we're moving from the baseline to the no-risk-sharing counterfactual) relative to the change in the consumption multiplier:\footnote{Based on align (\ref{eq:response_C}), the change in the log consumption multiplier can be decomposed as:
		$$\Delta \ln(m^c) = \Delta \ln(\beta^C) + \Delta \ln(m^Y),$$
		where the first term reflects the direct effect of risk sharing (the change in the pass-through) and the second term reflects the indirect effect (the change in the output multiplier).}
	\begin{align}
	\mathcal{S}^{indirect}_{i,t} \approx \frac{\Delta \ln(m^Y)}{\Delta \ln(m^c)}			
	\end{align}
	Based on the results from the quantitative model, the indirect channel accounts for 47 percent of the total reduction in consumption volatility provided by risk sharing.
	
	\paragraph{Role of the Armington Elasticity} The Armington elasticity governs to what extent households and firms switch towards imports when an increase in military spending raises domestic prices, which shapes the size of the domestic feedback loop and the output multiplier. In our baseline, we estimate the Armington elasticity $\psi$ to match the size of the output multiplier. However, the literature has found a range of different estimates and we therefore check our quantitative predictions based on a wider range of values for $\psi$. Table \ref{tab:armington} considers a low value of $\psi = 0.76$ \citep[as estimated by][]{boehm2023long} and a high value of $\psi = 2$ \citep[as used by][]{nakamura2014fiscal}.\footnote{\cite{boehm2023long} estimate an elasticity of tariff-exclusive trade flows to tariff changes of -0.76 in the year following the tariff change, and potentially smaller (in absolute value) upon impact. They show that the negative of this elasticity maps into the elasticity of substitution across varieties in a Krugman trade model with monopolistic competition (see p.892). In our model with Armington trade, this corresponds to the elasticity of substitution across goods of different origin, $\psi$.} All other parameters are kept the same. 
	
	A low value of $\psi$ amplifies the local general equilibrium effects, especially in the absence of any risk sharing. Removing risk sharing makes the output multiplier go up by 140 percent (from 1.57 to 3.73) rather than 93 percent as we observed in the baseline case. Consequently, the increase in the consumption multiplier is also larger with indirect effects accounting for 54 percent of the total change. Conversely, a higher trade elasticity mutes the local feedback loop and removing risk sharing would raise the output multiplier by only 61 percent (from 1.01 to 1.62). The consumption multiplier would triple, with indirect effects accounting for 38 percent. 
	
	\begin{table}     
		\centering             
		\begin{threeparttable}
			\caption{\textsc{Effect of Armington elasticity on multipliers}}
			\label{tab:armington}
			\begin{tabular}{l d{1.2} d{1.2} d{1.2}}
				\toprule  
				\multicolumn{1}{l}{Metric} 
				& \multicolumn{1}{c}{$\psi=0.76$} 
				& \multicolumn{1}{c}{$\psi=1.17$}  
				& \multicolumn{1}{c}{$\psi=2.00$}  \\ \midrule 
	\textbf{Aggregate output multiplier} \\
	Baseline & 1.57 & 1.32 & 1.01 \\
	No risk sharing & 3.73 & 2.55 & 1.62 \\ [.3cm]
	\textbf{Per-capita consumption multiplier} \\
	Baseline & 0.76 & 0.63 & 0.48 \\
	No risk sharing & 3.73 & 2.55 & 1.62 \\ [.3cm]
	Share indirect effect & 0.54 & 0.47 & 0.38 \\
				\bottomrule
			\end{tabular}
			\begin{tablenotes}[flushleft]
				\item {\footnotesize \textit{Notes:} The table reports the multipliers for aggregate nominal GDP and per-capita nominal consumption for different values of the Armington elasticity. The baseline value is $\psi=1.17$. The share of the indirect effect corresponds to $\Delta \ln (m^Y)$ / $\Delta \ln (m^c)$, where $\Delta \ln (m^Y)$ corresponds to the log of the output multiplier without risk sharing less the log of the output multiplier with the observed level of risk sharing (baseline).}
		\end{tablenotes}
	\end{threeparttable}
\end{table}

\subsection{The Role of Individual Risk-Sharing Channels}
Having established the aggregate benefits of risk sharing, we now isolate the contribution of each channel. We do this by shutting down one channel at a time while keeping the others active at their baseline calibration. This exercise reveals which mechanisms are most critical for stabilization and highlights the interdependencies between the channels.

Figure \ref{fig:counterfactuals} presents the results. Each subplot corresponds to a counterfactual where one channel is eliminated. The top bars show the resulting changes in output and consumption multipliers, while the bottom bars illustrate how the remaining channels' coefficients adjust.
\begin{figure}[h!]
	\centering
	\includegraphics[width=1\linewidth]{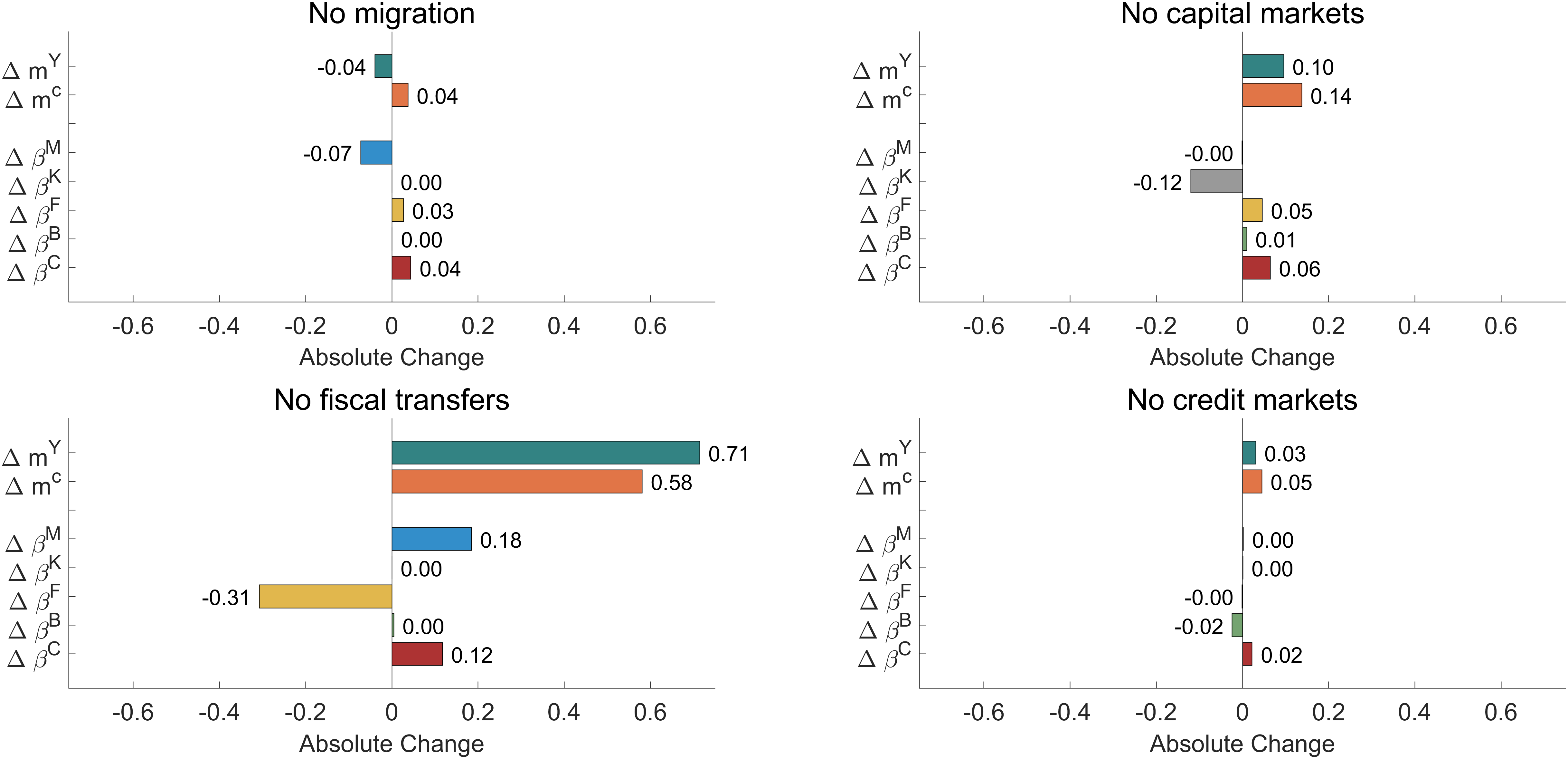}
	\caption{\textsc{Risk-sharing $\beta$'s for counterfactuals}}
	\label{fig:counterfactuals}	
	\begin{threeparttable}
		\begin{tablenotes}[]                                                                   
			\item {\footnotesize{\textit{Notes:} Risk-sharing $\beta$'s and multipliers are reported for alternative model calibrations. The notes to Figure \ref{fig:main_counterfactual} explain how the various $\beta$'s and the multipliers are estimated from the simulated data. In the counterfactuals, eliminating migration is achieved through $\gamma \rightarrow 0$, setting $\kappa=1$ eliminates cross-border capital markets, setting $\tau=0$ eliminates fiscal transfers, and setting $\iota \rightarrow \infty$ eliminates credit markets.}}
		\end{tablenotes}   
	\end{threeparttable}
\end{figure}

The figure reveals that fiscal transfers are the most important risk-sharing mechanism. Without the tax and the transfer system, output multipliers would increase by 0.71 and consumption multipliers by 0.58. Income diversification through capital markets also play some role, but its effect is more limited with output multipliers increasing by only 0.10 if the U.S. had no income diversification through capital markets. 

The bottom bars illustrate how some of the other risk-sharing channels pick up the slack when one channel is shut down. For instance, without the tax and transfer system, the fall in consumption during a local recession would be more pronounced, creating stronger incentives for workers to relocate to more prosperous regions. Quantitatively, these substitution effects are strong. In the baseline, 31 percent of local GDP fluctuations are absorbed by the tax and transfer system. Setting that to zero raises the share of unsmoothed fluctuations not by 31 percentage points, but only by 12 percentage points because the share absorbed by migration would increase from 7 to 25 percent. Similarly, shutting down income diversification through capital markets does not raise the share of unsmoothed fluctuations by 12 percentage points, but only by 6 percentage points because the strong fiscal transfer system would absorb almost half of these additional fluctuations. 

The substitutability across risk-sharing channels can also be seen from adding up the effects of the multiplier changes across the four subplots. If they operated independently, the sum of changes in individual consumption multipliers when shutting down each channel separately would equal the increase from shutting down all channels at once. Instead, we find that the sum of the four ``single-channel-off'' counterfactuals is 0.81 ($=0.04+0.14+0.58+0.05$, i.e. adding up the second-row bars across the four subplots), which is less than half of the total increase in consumption volatility of 1.91 ($=2.55-0.64$) when all channels are eliminated (see Figure \ref{fig:main_counterfactual}). This gap shows that when one risk-sharing mechanism is compromised, the remaining channels partially compensate.

The counterfactuals above quantify the marginal contribution of each channel around the observed equilibrium. However, because channels interact, these marginal effects are not constant. To obtain each channel's average marginal contribution across all possible combinations of active channels, we implement a Shapley decomposition that averages the effect of adding the channel over every subset of other channels.

\begin{table}     
	\centering             
	\begin{threeparttable}
		\caption[Shapley decomposition]{\textsc{Marginal effects of each risk-sharing channel}}
		\label{tab:shapley}
		\begin{tabular}{d{3.0} d{3.0} d{3.0} d{3.0} c d{5.0}}
			\toprule  
			\multicolumn{1}{c}{Migration} 
			& \multicolumn{1}{c}{Capital markets} 
			& \multicolumn{1}{c}{Fiscal transfers}  
			& \multicolumn{1}{c}{Credit markets} 
			&& \multicolumn{1}{c}{Sum}  \\ \midrule 
			16\% & 21\% & 55\% & 9\% && 100\%  \\ \bottomrule
		\end{tabular}
		\begin{tablenotes}[flushleft]
			\item {\footnotesize \textit{Notes:} This table reports the Shapley value for each risk-sharing channel's contribution to consumption smoothing. Each value is the average marginal contribution of a channel to reducing the consumption multiplier, computed across all possible coalitions of active channels. Values are expressed as shares of the total reduction in the consumption multiplier when moving from no risk sharing to the observed level of risk sharing.}
			
		\end{tablenotes}
	\end{threeparttable}
\end{table}

Table \ref{tab:shapley} presents these results, expressed as shares of the total reduction in consumption volatility. The main results from Figure \ref{fig:counterfactuals} survive: Fiscal transfers emerge as the most important mechanism, accounting for 55 percent of the total reduction in the consumption multiplier. Income diversification through capital markets contribute 21 percent, while migration and credit markets account for 16 and 9 percent, respectively.  

\subsection{Why Some States Benefit More from Risk Sharing Than Others}
Our final analysis explores why the benefits of risk sharing vary across states. For each state, we compute the output multiplier as the one-year response of GDP to a targeted increase in federal government spending equivalent to 1 percent of state GDP. We do this in both our baseline model and in the counterfactual with no risk sharing.

Figure \ref{fig:counterfactual_byState} displays the results. In the baseline scenario, multipliers are relatively homogeneous across states (ranging from 0.95 to 1.5). Without risk sharing, however, they become larger and more dispersed (ranging from 1.3 to 3.9). While the ranking of states is mostly unaffected, the model suggests that some states benefit substantially more from risk sharing than others. Panel (b) reveals a strong negative relationship between a state's import share and the ``benefit'' of risk sharing, measured as the increase in the output multiplier when risk sharing is shut down.

States with low import shares (i.e., relatively closed economies) like Hawaii and Colorado experience stronger feedback loops between consumption and income. In these states, risk-sharing mechanisms that weaken this feedback loop are particularly effective at reducing the multiplier. Conversely, for very open states like Kentucky or Vermont, a larger fraction of any demand shock naturally leaks out through imports, dampening the output multiplier even without formal risk-sharing mechanisms.

This finding highlights an interesting interaction between risk-sharing and trade openness. Trade openness was proposed by \cite{mckinnon1963optimum} as another precondition for an optimal currency area. We show that these two notions that have been discussed separately in the literature are partial substitutes.\footnote{An exception is \cite{farhi2014labor} that point out that labor mobility has a more stabilizing role in economies that are more open to trade.}

\begin{figure}[htbp]
	\centering
	\begin{subfigure}{0.48\linewidth}
		\includegraphics[width=\linewidth]{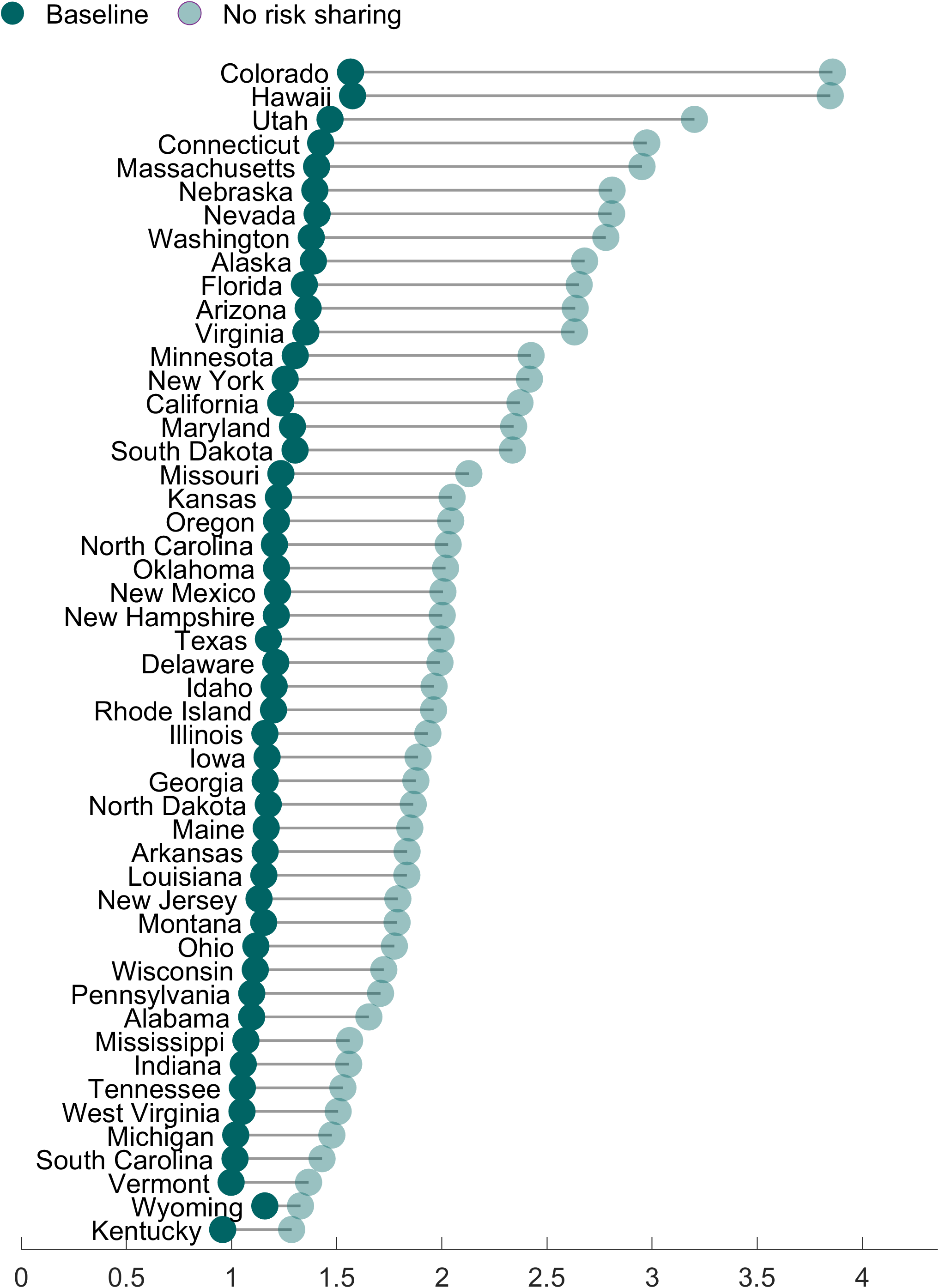}
		\caption{Output multipliers by state}
		\label{fig:counterfactual_byState}
	\end{subfigure}
	\hfill
	\begin{subfigure}{0.48\linewidth}
		\includegraphics[width=\linewidth]{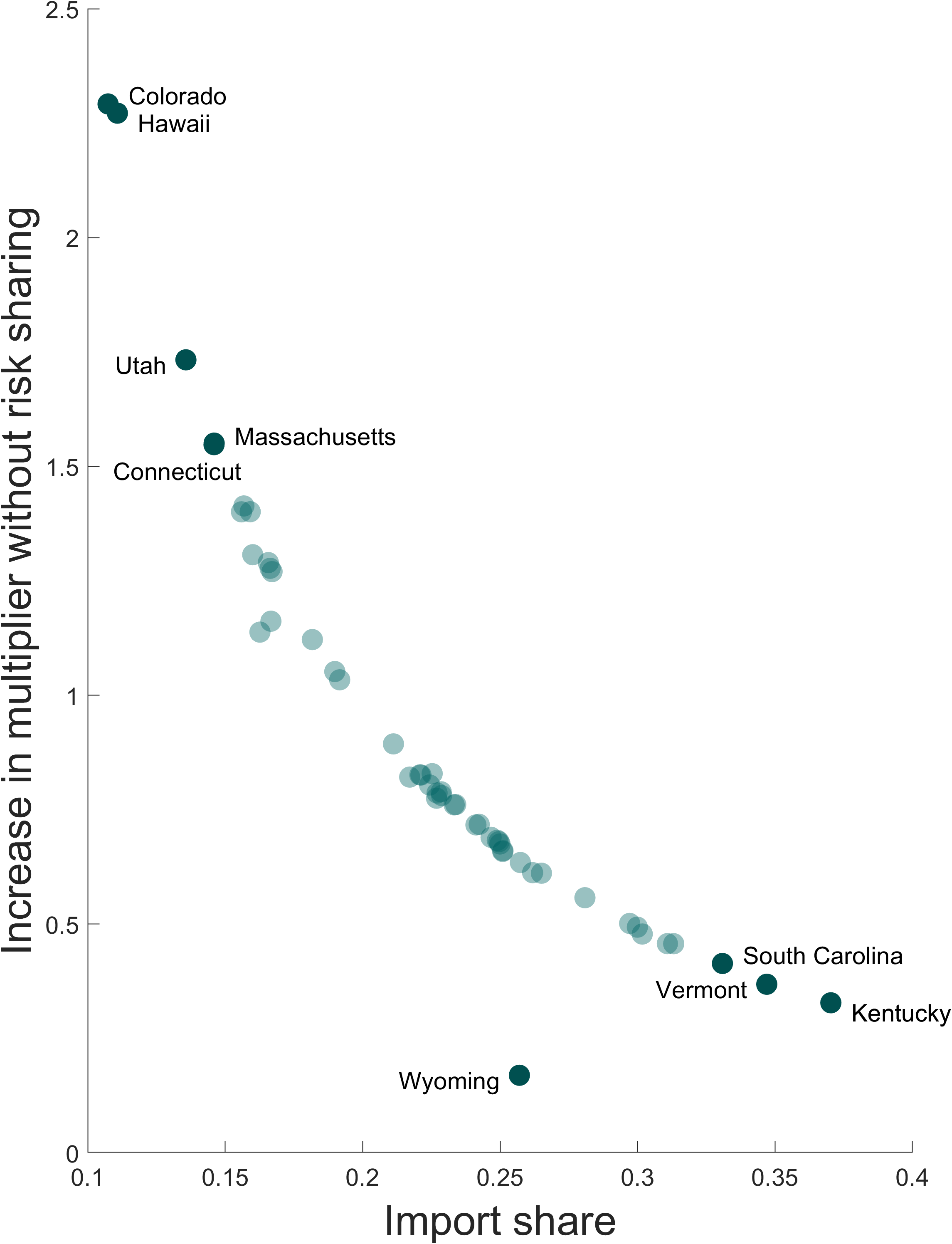}
		\caption{Openness and multipliers}
		\label{fig:openness_scatter}
	\end{subfigure}
	
	\caption{\textsc{Output multipliers and openness by state}}
	\begin{threeparttable}
		\begin{tablenotes}
			\item {\footnotesize \textit{Notes:} 
				The left panel displays output multipliers for each U.S. state in the baseline model and a model without any risk sharing. Multipliers are estimated from the model as the 1-year response of GDP to a one-time increase in the innovation to government spending that raises spending by a total of 1 percent of GDP over the first year in a particular state. 
				The right panel displays the change in the output multiplier from the baseline scenario to the no-risk-sharing scenario as a function of a state's home bias in trade (one minus the import share).}
		\end{tablenotes}
	\end{threeparttable}
\end{figure}

\section{Conclusion}
This paper re-evaluates the benefits of risk sharing in currency unions, showing that traditional risk-sharing measures underestimate these benefits by ignoring crucial general equilibrium effects on output volatility. Our central insight is that risk sharing not only directly smooths consumption for a given income shock but also indirectly dampens the income shock itself by stabilizing aggregate demand. These indirect benefits are quantitatively large. A multi-region DSGE model, disciplined by our empirical estimates, shows that the observed level of risk sharing across U.S. states reduces state-level consumption volatility by a factor of almost four, with close to half of this reduction attributable to these indirect, general equilibrium effects.

Although a full quantitative welfare analysis is outside the scope of this paper, our findings have important implications for understanding the welfare benefits of risk sharing.  First, by documenting how indirect effects substantially amplify the consumption-smoothing benefits of risk sharing, our results imply that the welfare gains from these mechanisms are considerably larger than previously thought. Second, our finding that risk sharing reduces output volatility is critical, as the literature has established that output volatility itself can generate first-order welfare costs. For instance, work by \cite{schmitt2016downward} and \cite{dupraz2025plucking} highlight how output volatility directly lowers average output through mechanisms such as downward nominal wage rigidity. Accounting for these first-order effects would further magnify the welfare gains from risk sharing beyond what consumption smoothing alone would imply.

We see an important contribution of our paper to analyze the benefits of risk sharing through the lens of a structural model, which allows us to conduct proper counterfactuals. However, there are some limitations to our exercise: For instance, we do not account for potential endogenous responses in production specialization. As \cite{kalemli2003risk} document, greater geographical diversification of income sources through capital markets can enable regions to specialize more intensively, since a higher variance in locally generated income need not translate into a higher variance in received income.\footnote{This is consistent with our finding that capital income is 2-3 times more volatile than labor income at the state level.} Our counterfactuals focus solely on the demand-side feedback effects of risk sharing on income volatility, abstracting from these potential supply-side adjustments.

While our analysis centers on U.S. states, our findings have important implications for other currency unions, particularly the euro area. The conventional wisdom, reflected in numerous policy discussions, is that the primary deficiency in euro area risk sharing is the lack of capital market integration compared to the United States \citep{nikolov2016cross,cimadomo2023changing}. Our results challenge this view. Both our causally identified empirical estimates and our structural model suggest that capital market integration may play a less important role in attenuating regional business cycles than previously thought. Instead, fiscal transfers emerge as a more potent stabilization mechanism. A fruitful exercise for future research would be to apply our integrated empirical and structural framework to the euro area to better gauge its progress toward becoming an ``optimal currency area'' and to identify the most effective policy interventions for enhancing macroeconomic stability.


\newpage 
\bibliographystyle{aea}
\bibliography{literature}



\end{document}